\documentclass[prd,tightenlines,nofootinbib,superscriptaddress]{revtex4}

\usepackage{amsfonts,amssymb,amsthm,bbm,amsmath}

\newcommand{\C}{{\mathbb C}}
\newcommand{\N}{{\mathbb N}}
\newcommand{\R}{{\mathbb R}}

\newcommand{\cA}{{\mathcal A}}
\newcommand{\cE}{{\mathcal E}}

\newcommand{\cJ}{{\mathcal J}}

\newcommand{\cH}{{\mathcal H}}

\newcommand{\cN}{{\mathcal N}}

\newcommand{\cC}{{\mathcal C}}
\newcommand{\cS}{{\mathcal S}}
\newcommand{\cP}{{\mathcal P}}
\newcommand{\cV}{{\mathcal V}}
\newcommand{\cR}{{\mathcal R}}

\newcommand{\SU}{\mathrm{SU}}
\renewcommand{\O}{\mathrm{O}}

\newcommand{\SL}{\mathrm{SL}}
\newcommand{\SO}{\mathrm{SO}}
\newcommand{\U}{\mathrm{U}}

\newcommand{\su}{{\mathfrak{su}}}

\renewcommand{\u}{{\mathfrak{u}}}
\renewcommand{\o}{{\mathfrak{o}}}

\newcommand{\be}{\begin{equation}}
\newcommand{\ee}{\end{equation}}
\newcommand{\beq}{\begin{eqnarray}}
\newcommand{\eeq}{\end{eqnarray}}
\newcommand{\bes}{\begin{eqnarray}}
\newcommand{\ees}{\end{eqnarray}}
\newcommand{\bea}{\begin{eqnarray}}
\newcommand{\eea}{\end{eqnarray}}
\newcommand{\nn}{\nonumber}

\newcommand{\mat} [2] {\left ( \begin{array}{#1}#2\end{array} \right ) }

\newcommand{\bin} [2] {\left (\begin{array}{c}#2\\#1\end{array} \right ) }

\newcommand{\la}{\langle}
\newcommand{\ra}{\rangle}

\newcommand{\w}{\wedge}

\newcommand{\tr}{{\mathrm Tr}}
\newcommand{\f}{\frac}

\def\nn{\nonumber}
\def\pp{\partial}
\def\arr{\rightarrow}

\newcommand{\id}{\mathbb{I}}

\def\vphi{\varphi}
\def\eps{\epsilon}
\def\om{\omega}
\def\Om{\Omega}

\def\bz{\bar{z}}
\def\bzeta{\bar{\zeta}}
\def\bF{\bar{F}}
\def\bE{\bar{E}}
\def\bU{\bar{U}}
\def\bbeta{\bar{\beta}}
\def\bw{\bar{w}}
\def\vcC{\vec{\cC}}
\def\vsigma{\vec{\sigma}}
\def\vu{\vec{u}}
\def\vn{\vec{n}}
\def\vv{\vec{v}}
\def\hv{\widehat{v}}
\def\vV{\vec{V}}
\def\tI{\tilde{I}}
\def\tz{\tilde{z}}
\def\btz{\bar{\tilde{z}}}
\def\tlambda{\tilde{\lambda}}
\def\tX{\widetilde{X}}
\def\tV{\widetilde{V}}

\def\he{\hat{e}}

\def\wg{\textrm{Wg}}

\newtheorem{theorem}{Theorem}[section]
\newtheorem{lemma}[theorem]{Lemma}
\newtheorem{prop}[theorem]{Proposition}
\newtheorem{res}[theorem]{Result}

\newtheorem{definition}[theorem]{Definition}

\begin{document}

\title{Deformations of Polyhedra and Polygons by the Unitary Group}

\author{{\bf Etera R. Livine}}\email{etera.livine@ens-lyon.fr}
\affiliation{Laboratoire de Physique, ENS Lyon, CNRS-UMR 5672, 46 All\'ee d'Italie, Lyon 69007, France}
\affiliation{Perimeter Institute, 31 Caroline St N, Waterloo ON, Canada N2L 2Y5}

\date{\today}

\begin{abstract}

We introduce the set of framed (convex) polyhedra with $N$ faces as the symplectic quotient $\C^{2N}//\SU(2)$. A framed polyhedron is then parametrized by $N$ spinors living in $\C^{2}$ satisfying suitable closure constraints and defines a usual convex polyhedron plus extra $\U(1)$ phases attached to each face. We show that there is a natural action of the unitary group $\U(N)$ on this phase space, which changes the shape of faces and allows to map any (framed) polyhedron onto any other with the same total (boundary) area. This identifies the space of framed polyhedra to the Grassmannian space $\U(N)/\,(\SU(2)\times\U(N-2))$.
We show how to write averages of geometrical observables (polynomials in the faces' area and the angles between them)  over the ensemble of polyhedra (distributed uniformly with respect to the Haar measure on $\U(N)$) as polynomial integrals over the unitary group and we provide a few methods to compute these integrals systematically. We also use the 
Itzykson-Zuber formula from matrix models as the generating function for these averages and correlations.

In the quantum case, a canonical quantization of the framed polyhedron phase space leads to the Hilbert space of $\SU(2)$ intertwiners (or, in other words, $\SU(2)$-invariant states in tensor products of irreducible representations). The total boundary area as well as the individual face areas are quantized as half-integers (spins), and the Hilbert spaces for fixed total area form irreducible representations of $\U(N)$. We define semi-classical coherent intertwiner states peaked on classical framed polyhedra and transforming consistently under $\U(N)$ transformations. And we show how the $\U(N)$ character formula for unitary transformations is to be considered as an extension of the Itzykson-Zuber to the quantum level and generates the traces of all polynomial observables over the Hilbert space of intertwiners.

We finally apply the same formalism to two dimensions and show that classical (convex) polygons can be described in a similar fashion trading the unitary group for the orthogonal group. We conclude with a discussion of the possible (deformation) dynamics that one can define on the space of polygons or polyhedra.
This work is a priori useful in the context of discrete geometry but it should hopefully also be relevant to (loop) quantum gravity in 2+1 and 3+1 dimensions when the quantum geometry is defined in terms of gluing of (quantized) polygons and polyhedra.

\end{abstract}

\maketitle


\section{Introduction}

Inspired by loop quantum gravity \cite{lqg_review}, and more particularly the spinorial formalism \cite{spinor1,spinor2,spinor3}  and the structures of twisted geometry \cite{twisted1}, we discuss the phase space of polyhedra in three dimensions and its quantization, which serves as basic building of the kinematical states of discrete geometry. More precisely, following \cite{un1}, we show that the Grassmannian space $\U(N)/(\U(N-2)\times \SU(2))$ is the space of framed (convex) polyhedra with $N$ faces up to 3d rotations. The framing consists in the additional information of a $\U(1)$ phase per face. This provides an extension of the Kapovich-Milson phase space \cite{KM} for polyhedra with fixed number of faces and fixed areas for each face. Indeed, we describe the Grassmannian as the symplectic quotient $\C^{2N}//\SU(2)$, which provides canonical complex variables for the Poisson bracket. This construction allows a natural $\U(N)$ action on the space of polyhedra, which has two main features. First, $\U(N)$ transformations act non-trivially on polyhedra and change the area and shape of each individual face. Second, this action is cyclic: it allows to go between any two polyhedra with fixed total area (sum of the areas of the faces) and in particular to generate any polyhedron from the totally squeezed polyhedron with only two non-trivial faces.

Upon quantization, the framed polyhedron phase space leads to the Hilbert space of $\SU(2)$ intertwiners, which is interpreted as the space of quantum polyhedra. We perform a canonical quantization from the complex variables of $\C^{2N}//\SU(2)$ and all the classical features are automatically exported to the quantum level. Each face carries now a irreducible representation of $\SU(2)$, i.e. a half-integer spin $j$, which defines the area of the face. Intertwiners are then $\SU(2)$-invariant states in the tensor product of these irreducible representations. These intertwiners are the basic building block of the spin network states of quantum geometry in loop quantum gravity. The $\U(N)$ action on the space of intertwiners changes the spins of the faces and each Hilbert space for fixed total area (sum of the spins) defines an irreducible representation of the unitary group $\U(N)$, as shown in \cite{un1}. Once again, the $\U(N)$ action is cyclic and allows to generate the whole Hilbert space from the action of $\U(N)$ transformation on the highest weight vector. This construction provides coherent intertwiner states peaked on classical polyhedra, as used in \cite{un2}.

\medskip

At the classical level, we will use the $\U(N)$ structure of the space of polyhedra to compute the averages of polynomial observables over the ensemble of polyhedra distributed along the uniform Haar measure. We will underline a phenomenon of concentration of measure, which peaks random polyhedra on spherical configurations for large number of faces $N$. Furthermore, we will show how to use the Itzykson-Zuber formula from matrix models \cite{iz1} as a generating functional for these averages. It computes the integral over $\U(N)$ of the exponential of the matrix elements of a unitary matrix tensor its complex conjugate.
At the quantum level, we will show that the character formula, giving the trace of unitary transformations either over the standard basis or the coherent intertwiner basis, provides an extension of the Itzykson-Zuber formula. It allows in principle to generate the expectation values of all polynomial observables (and thus their spectrum).

\medskip

The plan of this paper goes as follows. In section II, we define and describe the phase space of framed polyhedra, its parameterization in terms of spinor variables and the action of $\U(N)$ transformations. In section III, we show how to compute the averages and correlations of polynomial observables using group integrals over $\U(N)$ and we discuss the Itzykson-Zuber integral as generating function. In section IV, we discuss the quantum case, with the Hilbert space of $\SU(2)$ intertwiners, coherent states and the character formula. In section V, we investigate the lower-dimensional analog of polygons (in two dimensions), we show that the unitary group is replaced by the orthogonal group and that the Grassmannian $\O(N)/(\O(N-2)\times\SO(2))$ defines the phase space for framed polygons. We then discuss the issue of gluing such polygons together into a consistent 2d cellular decomposition, as a toy model for the gluing of framed polyhedra into 3d discrete manifolds.


\medskip

These constructions are relevant to quantum gravity in 2+1 and 3+1 dimensions, especially to discrete approaches based on a description of the geometry using glued polygons and polyhedra such as loop quantum gravity (and dynamical triangulations). The goal is to clarify how to parametrize the set of polygons/polyhedra and their deformations, and to introduce mathematical tools to compute the average and correlations of observables over the ensemble of polygons/polyhedra at the classical level and then the spectrum and expectation values of geometrical operators on the space of quantum polygons/polyhedra at the quantum level.

In this context, we hope that this work will be useful to the study of the dynamics of (loop) quantum gravity, especially in its formulation in terms of spinor networks and twisted geometries, but it should also be relevant to the study of the structure of discrete geometries and cellular decompositions.

\section{Phase Space of Polyhedra and Unitary Group Action}

\subsection{A Quick Review of the Kapovich-Milson Phase Space}

Let us consider $N$ vectors $\vV_i$ in $\R^3$ that satisfy a closure condition, that their sum vanishes:
\be
\sum_{i=1}^N \vV_i \,=\,0.
\ee
By a theorem due to Minkowski, these determine a unique convex polyhedron with $N$ faces, such that the $\vV_i$'s are the (outward) normal vectors to the faces, that is the faces have area $V_i=|\vV_i|\in\R_+$ and unit normal  $\hv_i=\vV_i/|\vV_i|\in\cS^2$. The reconstruction of the polyhedron is not trivial and the shape of the faces depend non-trivially on the set of chosen vectors. The interested reader can find details on the reconstruction algorithm in \cite{dona}.

The space of polyhedra $\cP_N=\{(\vv_i)\, |\, \sum_{i=1}^N \vV_i =0\}$ has dimension $(3N-3)$, and if we consider the set of equivalence classes under 3d rotations we get the space $\cP_N^0$ with dimension  $(3N-6)$. Generally, these spaces do not have an even dimension and are not symplectic manifolds. However, if we fix the areas $V_i$ of all $N$ faces, we get the Kapovich-Milson phase space \cite{KM}:

\begin{definition}

Let us consider the space product of $N$ 2-spheres for fixed $V_i$'s:
\be
\cS_N^{\{V_i\}}
\,\equiv\,
\{(\vV_i)_{i=1..N}\,\in\,(\R^3)^N\, s.t.\, |\vV_i|=V_i\}
\,\sim\,
\{(\hv_i)_{i=1..N}\,\in\,(\cS^2)^N\}\,.
\ee
This is a symplectic manifold provided with the Poisson structure on each of the $N$ spheres (scaled by their radii):
\be
\{\cdot,\cdot\}
\,=\,
\sum_i\eps^{abc}V_i^c\f{\pp\,\cdot}{\pp V_i^a}\f{\pp\,\cdot}{\pp V_i^b},
\qquad
\{V_i^a,V_i^b\}=2\eps^{abc}V_i^c,
\qquad
\{V_i,V_i^a\}0,
\qquad
\{v_i^a,v_i^b\}=\,\f1{V_i}2\eps^{abc}v_i^c\,.
\ee
Then the closure conditions $\sum_iV_i^a = \sum_i V_iv_i^a=0$ form a first class constraint system, that generates global $\SO(3)$ rotations on the set of the $N$ vectors $\vV_i$, or equivalently on the set of the $N$ unit vectors $\hv_i$.
This defines by symplectic reduction the Kapovich-Milson phase space for convex polyhedra with $N$ faces and fixed face areas $V_i$:
\be
\cP_N^{\{V_i\}}
\,\equiv\,
\cS_N^{\{V_i\}}//\SO(3)
\,=\,
\{(\vv_i)_{i=1..N}\,\in\,(\cS^2)^N\, s.t.\, |\vV_i|=V_i\}/\SO(3)\,.
\ee
This manifold has dimension $(2N-6)$.

\end{definition}

Instead of removing $N$ degrees of freedom from the space of polyhedra $\cP_N^r$ by fixing the individual face areas and thus obtaining the manifold $\cP_N^{\{V_i\}}$ with even dimension $(3N-6)-N$ and carrying a symplectic structure, we will now add $N$ degrees of freedom to embed $\cP_N^0$ into a larger phase space of framed polyhedra $\cP_N^z$ with even dimension $(3N-6)+N$. These extra degrees of freedom are angles (or $\U(1)$ phases) canonically conjugate to the face areas. They allow to work in a phase space where the areas can vary and have a dynamics. This is a necessary structure when studying the dynamics of loop quantum gravity, where areas and spins do change under time evolution (and space-time diffeomorphisms).
We achieve this below by using the spinorial representation of the $\su(2)$ algebra \`a la Schwinger as prescribed in \cite{un1,un2,spinor1,spinor2,spinor3,twisted1}.

\subsection{Spinor Phase Space for Framed Polyhedra}

We will now replace the data of $N$ vectors in $\R^3$ by $N$ spinors.
We call a spinor a complex 2-vector $z\in C^2$ for which we will use a bra-ket notation:
$$
|z\ra =\mat{c}{z^0 \\ z^1},
\qquad
\la z| =\mat{cc}{\bz^0 \\ \bz^1}\,.
$$
It lives in the fundamental 2-dimensional representation of $\SU(2)$, with the obvious scalar product $\la w|z\ra= \bw^0 z^0+\bw^1 z^1$. We also introduce its dual spinor using the structure map of $\SU(2)$:
$$
|z]=\eps |\bz\ra =\mat{c}{-\bz^1 \\ \bz^0},
\qquad
[ z| =\mat{cc}{-z^1 \\ z^0}\,.
$$

Following \cite{un1,un2}, we consider sets of $N$ spinors satisfying a closure constraint:

\begin{definition}

Let us consider the space $\C^{2N}$  of $N$ spinors $z_i\in\C^2$ endowed with the canonical symplectic structure $\{z^A_i,\bz^B_j\}=-i\,\delta^{AB}\,\delta_{ij}$ with the indices $A,B=0,1$.
We impose the closure constraints that the 2$\times$ 2 matrix $X\,\equiv\,\sum_i |z_i\ra\la z_i|$ is proportional to the identity:
\be
\vcC
\equiv
\tr X\vsigma
\,=\,
\sum_i \la z_i|\vsigma|z_i\ra
\,=\,
0\,,
\qquad\textrm{or equivalently}\quad
\sum_i |z_i\ra\la z_i|
\,=\,
\f12\,
\sum_i \la z_i|z_i\ra
\,\id\,,
\ee
$$
\textrm{or explicitly}\qquad
\sum_i|\tz_i^0|^2-|\tz_i^1|^2=0
\quad\textrm{and}\quad
\sum \btz_i^0\tz_i^1=0\,,
$$
where the three matrix $\sigma_{a=1,2,3}$ are the Pauli matrices generating $\SU(2)$.
These three real constraints are first class and generate the $\SU(2)$ action on the $N$ spinors:
$$
\{\vcC\,,\,|z_i\ra\}
\,=\,
i\,\vsigma\,|z_i\ra,
\qquad
e^{\{\vu\cdot\vcC,\cdot\}}\,|z_i\ra
\,=\,
g\,|z_i\ra\,,
\quad
e^{\{\vu\cdot\vcC,\cdot\}}\,|z_i]
\,=\,
g\,|z_i]\,,
\quad
g=e^{i\,\vu\cdot\vsigma}\in\SU(2)\,.
$$
We define the phase space of framed polyhedra with $N$ faces as the symplectic quotient $\cP_N^z\equiv \C^{2N}//\SU(2)$, that is as the set of collections of $N$ spinors satisfying the closure constraints and up to $\SU(2)$ transformations.

\end{definition}

A simple counting gives that $\cP_N^z$ is a $(4N-6)$-dimensional manifold, which corresponds to the dimension $(3N-6)$ of the space $\cP_N^0$ of $N$-faced polyhedra (up to 3d rotations) plus $N$ degrees of freedom.

More precisely, we introduce the mapping from spinors to 3-vectors:
\be
|z\ra\in\C^2
\quad\longmapsto\quad
\vV\equiv\la z|\vsigma|z\ra\,\in\R^3\,,
\ee
with $V=|\vV|=\la z|z\ra$.
This mapping is obviously not one-to-one and is actually invariant under the multiplication of the spinor by an arbitrary phase, $|z\ra\,\rightarrow\,e^{i\theta}\,|z\ra$. The inverse mapping is given by \cite{spinor1,un3}:
\be
|z\ra \,=\,
e^{i\theta}\,\f1{\sqrt{2}}\,\,
\mat{c}{\sqrt{V+V_z}\\ e^{i\vphi}\sqrt{V-V_z}},
\qquad\textrm{with}\quad
e^{i\vphi}=\f{V_x+iV_y}{\sqrt{V_x^2+V_y^2}}=\f{V_x+iV_y}{\sqrt{V^2-V_z^2}}\,.
\ee
This provides a bijection $\C^2\sim \R^3\times\U(1)$. One checks that we have the same Poisson brackets for the vectors as earlier, $\{V^a,V^b\}=2\eps^{abc}V^c$, inherited from the canonical bracket on the spinor variables.

Using this mapping, we send a collection of $N$ spinors onto a collection of $N$ vectors. The closure constraint then read as before:
$$
\vcC = \sum_i \vV_i =0\,.
$$
This defines a convex polyhedron with $N$ faces with areas given by the norm squared of the spinors $|V_i|=\la z_i|z_i\ra$, with a total area $\cA=2\lambda\equiv\sum_i |V_i|=\sum_i \la z_i|z_i\ra$ overall.

This mapping provides a bijection $\cP_N^z\sim \cP_N^0 \times\U(1)^N$ between our space of framed polyhedra defined in terms of spinors and the space $\cP_N^0$ of polyhedra with $N$ faces up to 3d rotations times $N$ phases attached to each face \cite{un1,un2}. This construction provides a larger phase space where the areas of the faces can vary dynamically.
Moreover the spinors are crucial in defining the action of the unitary group $\U(N)$ on the (framed) polyhedra as we will see in the next sections.

\medskip

Finally, we conclude this section by introducing the complex variable $\zeta=z^1/z^0\in\C$ for a spinor $|z\ra$. This variable $\zeta$ commutes with the norm $V$ and parameterizes the 2-sphere defined by the 3-vector $\vV$ as $|z\ra$ varies while keeping the radius $V=\la z|z\ra$ fixed:
\be
\{V,\zeta\}\,=\,0,\qquad
\vV=\,V\,\hv \quad\textrm{with}\quad
v_z=\f{1-|\zeta|^2}{1+|\zeta|^2},\quad
v_+=\f{\zeta}{1+|\zeta|^2}\,.
\ee
The symplectic structure on the 2-sphere then simply reads in terms of this complex parameter:
\be
\{\zeta,\bzeta\}
\,=\,
\f{-i}{V}\,\big{(}1+|\zeta|^2\big{)}\,.
\ee
This variable is specially interesting when studying the Kapovich-Milson phase space for fixed individual face ares, when the phase space is parametrized by these complex variables $\zeta_{i=1..N}$ constrained by the closure condition.

\subsection{Closing Open Polyhedra and the $\SL(2,\C)$ Action}

A first interesting remark is that the use of spinors provide a natural way to close opened configurations into actual polyhedra.
As pointed out in \cite{un2,un3,un4}, this is achieved through a $\SL(2,\C)$ transformation on the spinors.

Indeed, starting with an arbitrary set of (not all vanishing) $N$ spinors $|z_i\ra$, a priori not satisfying the closure constraints, that is such that the matrix $X=\sum_i |z_i\ra\la z_i|$ is not proportional to the identity. Then $X$ is a positive Hermitian operator, it can be diagonalized and written as:
\be
X=\sum_i |z_i\ra\la z_i|
= g\Delta g^{-1}
=\rho \Lambda\Lambda^\dagger,
\ee
where  $g\in\SU(2)$ is unitary,  $\Delta$ is a diagonal 2$\times$2 matrix with positive entries, $\rho=\det X=\det \Delta$ is positive, $\Lambda=g\sqrt{\Delta}/\rho^{\f14}\in\SL(2,\C)$. Then we act with $\Lambda^{-1}\in\SL(2,\C)$ on the spinors to get a closed configuration:
\be
|z_i\ra
\,\overset{\Lambda^{-1}}{\longrightarrow}\,
|\tz_i\ra
\,\equiv\,
\Lambda^{-1}|z_i\ra\,.
\ee
These new spinors $|\tz_i\ra$ trivially satisfy the closure constraints:
$$
\tX
=\sum_i |\tz_i\ra\la \tz_i|
=\Lambda^{-1}X(\Lambda^\dagger)^{-1}
=\rho\id\,,
$$
and thus define a (framed) polyhedron with face areas $\tV_i=\la \tz_i|\tz_i\ra$ and total area:
\beq
&&
2\tlambda
=\sum_i \tV_i= \tr \tX= 2\rho,\\
&&\rho^2=\det X
=\f12\big{[}(\tr X)^2-\tr X^2\big{]}
=\f14\big{[}(\tr X)^2-\tr (X\vsigma)\cdot\tr (X\vsigma)\big{]}
=\f14\left[(2\lambda)^2-|\vcC|^2\right].\nn
\eeq
This new total area $2\tlambda$ is always smaller than the initial one $2\lambda$ and obviously coincides when the original spinors already satisfy the closure condition $\vcC=0$.

\medskip

It is useful to get a closer at the geometry of this procedure. Starting with the $N$ vectors $\vV_i$ with a non-vanishing sum $\vcC\ne 0$, we perform a $\SU(2)$ transformation $g$ on the spinors $|z_i\ra$ such that the corresponding $\S0(3)$ rotation sends the vector $\vcC$ onto the $z$-axis:
$$
\vcC\,=\,|\vcC|\,g\vartriangleright \hat{e}_z\,,
$$
where $\hat{e}_z$ is the unit basis vector along $z$-axis.
Up to this 3d rotation, we can start directly with such a configuration with $\vcC$ collinear with the $z$-axis. Then writing the components of the matrix $X$ in terms of the spinors give equations corresponding to the total area and the components of the closure vector:
\be
\tr X=\sum_i |z_i^0|^2+|z_i^1|^2=2\lambda,
\quad
\tr X\sigma_z= \cC^z= \sum_i |z_i^0|^2-|z_i^1|^2=|\vcC|,
\quad
\tr X\sigma_+=\cC^+=\sum \bz_i^0z_i^1=0\,.
\ee
We now defines the rescaled spinors:
\be
|z_i\ra \,\arr\,
|\tz_i\ra=\,\mat{cc}{\mu & 0 \\ 0 & \mu^{-1}}\,|z_i\ra
=\,\Lambda^{-1}\,|z_i\ra
\quad \textrm{with}\,
\mu=\sqrt{\f{\lambda-\f{|\vcC|}2}{\lambda+\f{|\vcC|}2}}\,,
\ee
or explicitly:
$$
z_i^0\,\arr\,\tz_i^0=\sqrt{\f{\lambda-\f{|\vcC|}2}{\lambda+\f{|\vcC|}2}}z_i^0,
\quad
z_i^1\,\arr\,\tz_i^1=\sqrt{\f{\lambda+\f{|\vcC|}2}{\lambda-\f{|\vcC|}2}}z_i^1.
$$
First, the new spinors $|\tz_i\ra$ satisfy the balance equation $\sum_i |\tz_i^0|^2=\sum_i|\tz_i^1|^2$ and the orthogonality equation $\sum \btz_i^0\tz_i^1=0$, and thus satisfy the closure condition.
They define a closed polyhedron with total area $2\tlambda=\sum_i |\tz_i^0|^2+|\tz_i^1|^2= \sqrt{4\lambda^2-{|\vcC|^2}}$.

Second, the rescaling matrix $\Lambda$ is in $\SL(2,\C)$ and is actually a boost along the $z$-direction. And we understand the overall $\SL(2,\C)$ transformation from the original arbitrary spinors $|z_i\ra$ to the new closed spinors $|\tz_i\ra$ as a rotation to the $z$-axis followed by a rescaling of the first and second component of the spinors with inverse factor so that the sum of their modulus square match.

\subsection{Invariant Parametrization and Cross-Ratios}
\label{param}


The spinors $z_i^A$ do not commute with the closure constraints $\vcC=0$ and are thus not invariant under $\SU(2)$ transformations. The first question is to identify $\SU(2)$-invariant observables, which can then be used to parameterize framed polyhedra in the  phase space $\cP_N^z$.

Natural observables are given by the scalar products between spinors and their dual:
\be
E_{ij}=\la z_i|z_j\ra\,\quad
\bE_{ij}=E_{ji},\qquad
F_{ij}=[z_i|z_j\ra = -[z_j|z_i\ra,
\quad
\bF_{ij}=\la z_j|z_i] = -\la z_i|z_j]\,,
\ee
These scalar products commute with the closure constraints,
$$
\{\vcC,E_{ij}\}=\{\vcC,F_{ij}\}=\{\vcC,\bF_{ij}\}=0
$$
and are thus invariant under $\SU(2)$ transformations of the spinors,
$$
|z_i\ra\,\overset{g\in\SU(2)}{\longrightarrow}\,g\,|z_i\ra\,,
\quad
|z_i]\,{\longrightarrow}\,g\,|z_i]\,,
\quad
E_{ij}=\la z_i|z_j\ra\,{\longrightarrow}\,\la z_i|g^\dagger g|z_j\ra=\la z_i|z_j\ra,
\quad
F_{ij}=[z_i|z_j\ra\,{\longrightarrow}\,[z_i|g^\dagger g|z_j\ra=[ z_i|z_j\ra\,.
$$
These are the basic variables for the $\U(N)$ formalism for $\SU(2)$ intertwiners as developed for loop quantum gravity in \cite{un1,un2,un3,un4,un0}. From that perspective, the most useful feature is that these variables form a closed algebra under the Poisson bracket,
\bes
{\{}E_{ij},E_{kl}\}&=&
-i\left(\delta_{kj}E_{il}-\delta_{il}E_{kj} \right)\label{E_un}\\
{\{}E_{ij},F_{kl}\} &=& -i\left(\delta_{il}F_{jk}-\delta_{ik}F_{jl}\right),\qquad
{\{}E_{ij},\bF_{kl}\} = -i\left(\delta_{jk}\bF_{il}-\delta_{jl}\bF_{ik}\right), \\
{\{} F_{ij},\bF_{kl}\}&=& -i\left(\delta_{ik}E_{lj}-\delta_{il}E_{kj} -\delta_{jk}E_{li}+\delta_{kl}E_{li}\right), \nn\\
{\{} F_{ij},F_{kl}\} &=& 0,\qquad {\{} \bF_{ij},\bF_{kl}\} =0.\nn
\ees
This algebra will get quantized exactly and will provide the basic operators acting on the Hilbert space of intertwiners.

The usual vector scalar products $\vV_i\cdot\vV_j$, measuring the angles between two faces, are easily expressed in terms of these variables,
\be
\vV_i\cdot\vV_i=V_i^2=\la z_i|z_i\ra^2=E_{ii}^2,\quad
\vV_i\cdot\vV_j
\,=\,
2|E_{ij}|^2-V_iV_j
\,=\,
-2|F_{ij}|^2+V_iV_j\,.
\ee
One can write all observables probing the geometric of the polyhedra in terms of $E$'s or $F$'s. We can then use these variables to parameterize the space of (framed) polyhedra. On the one hand, the $E$'s are most particularly relevant because they generate $\U(N)$ transformations compatible with the closure conditions on the spinors. We will use this to define the action of the unitary group $\U(N)$ on polyhedra in the next section. On the other hand, the $F$'s are holomorphic and offer a enlightening parametrization of the framed polyhedron phase space as we explain below. Moreover, they are crucial in defining coherent intertwiners \cite{un2} and in deriving the holomorhic/anti-holomorphic splitting of the simplicity (second class) constraints in loop quantum gravity \cite{un3,simplicityL}.

The $F$'s are specially interesting because they are not only invariant under global $\SU(2)$ transformations but they are also invariant under global $\SL(2,\C)$ transformations, as it is easy to check:
$$
|z_i\ra\,\overset{\Lambda\in\SL(2,\C)}{\longrightarrow}\,\Lambda\,|z_i\ra\,,
\quad
|z_i]\,{\longrightarrow}\,\eps\bar{\Lambda}\eps^{-1}\,|z_i]=(\Lambda^{-1})^\dagger\,|z_i]\,,
\quad
F_{ij}=[z_i|z_j\ra\,{\longrightarrow}\,[z_i|\Lambda^{-1}\Lambda|z_j\ra=[ z_i|z_j\ra\,.
$$
Thus the action of closing an arbitrary set of spinors into a (framed) polyhedron, as described in the previous section, will leave the $F$'s invariant. We can go further and show that the $F$'s entirely determine the orbit under $\SL(2,\C)$ in the space of unconstrained spinors $\C^{2N}$:

\begin{lemma}
Considering two sets of spinors $|z_i\ra$ and $|w_i\ra$ such that $[z_i|z_j\ra=[w_i|w_j\ra$ for all indices $i,j$, and further assuming that there exists a couple of indices $k,l$ such that $[z_k|z_l\ra\ne0$, then there exists a matrix $\Lambda\in\SL(2,\C)$ that maps one onto the other:
\be
\forall i,j,\quad [z_i|z_j\ra=[w_i|w_j\ra
\qquad\Rightarrow\qquad
\exists \,\Lambda\in\SL(2,\C),\quad
\forall i,\,\, |z_i\ra=\Lambda\,|w_i\ra\,.
\ee
\end{lemma}
\begin{proof}

Let us first remark that the following identity on 2$\times$2 matrices is true, taking into account that $[z_k|z_l\ra\ne0$:
\be
\f{|z_l\ra [z_k|-|z_k\ra [z_l|}{[z_k|z_l\ra}
\,=\,\id_2\,.
\ee
Indeed, $[z_k|z_l\ra\ne0$ implies that $|z_k\ra$ and $|z_l\ra$ are not colinear and span the whole two-dimensional spinor space. Then the previous operator leaves invariant $|z_k\ra$ and $|z_l\ra$ and is thus equal to the identity.
Let us now consider the matrix:
\be
\Lambda\,\equiv\,
\f{|z_l\ra [w_k|-|z_k\ra [w_l|}{[w_k|w_l\ra}\,.
\ee
One checks that its determinant is equal to one, $\det \Lambda=\f12((\tr\Lambda)^2-\tr\Lambda^2) =1$, so that $\Lambda\in\SL(2,\C)$. Finally, using the equality of the $F$-observables for both sets of spinors, we have:
$$
\forall i,\quad
\Lambda\,|w_i\ra
\,=\,
\f{|z_l\ra [w_k|w_i\ra-|z_k\ra [w_l|w_i\ra}{[w_k|w_l\ra}
\,=\,
\f{|z_l\ra [z_k|z_i\ra-|z_k\ra [z_l|z_i\ra}{[z_k|z_l\ra}
\,=\,
|z_i\ra\,.
$$

\end{proof}

Furthermore each $\SL(2,\C)$-orbit has a unique intersection with the space of framed polyhedra $\cP^z_N$.  This is a re-statement of the isomorphism $\C^{2N}/\SL(2,\C)\sim \C^{2N}//\SU(2)$, where $\SL(2,\C)$ is understood as the complexification of $\SU(2)$. This is similar to the analysis performed in \cite{closure} but the present setting is slightly more general (and actually simpler) since the authors were looking at the Kapovich-Milson phase spaces (at fixed individual face areas). We formalize this as follows:

\begin{prop}

Considering two sets of spinors $|z_i\ra$ and $|w_i\ra$ satisfying the closure constraints, and such that $[z_i|z_j\ra=[w_i|w_j\ra$ for all indices $i,j$, then they are related by a global $\SU(2)$ transformation that maps one set of spinors onto the other:
\be
\exists \,g\in\SU(2),\quad
\forall i,\,\, |z_i\ra=g\,|w_i\ra\,.
\ee

\end{prop}

\begin{proof}

To start with, assuming the closure constraints on the spinors $z_i$, one can get the total area $2\lambda=\sum_i V_i=\sum_i \la z_i|z_i\ra$ from the $F$'s:
\be
\sum_{i,j} |F^{(z)}_{ij}|^2
\,=\,
\sum_{i,j} [z_i|z_j\ra\la z_j|z_i]
\,=\,
\tr (\lambda \id)^2
\,=\,
2\lambda^2\,.
\ee
Thus the total area associated to both sets of spinors $z_i$ and $w_i$ are equal.
If $\lambda$ vanishes, then both sets of spinors vanish and are trivially related by an arbitrary $\SU(2)$ transformation. Else $\lambda$ does not vanish and there automatically exists at least a couple of indices $(k,l)$ such that $F^{(z)}_{kl}=[z_k|z_l\ra$ does not vanish, so that we can apply the previous lemma ensuring that both sets of spinors are related by a $\SL(2,\C)$ transformation. Then the work is to show that this $\SL(2,\C)$ transformation is actually unitary and lays in $\SU(2)$.

First, we show that all the scalar products are equal, by inserting the closure constraint:
\be
\forall i,j,\quad
\la z_i|z_j\ra
\,=\,
\f1\lambda \sum_m \la z_i|z_m][z_m|z_j\ra
\,=\,
\f1\lambda \sum_m \la w_i|w_m][w_m|w_j\ra
\,=\,
\la w_i|w_j\ra\,.
\ee

Then we fix one index $k$ and consider the $\SU(2)$ group element mapping $w_k$ to $z_k$:
\be
g_k\,\equiv\,
\f{|z_k\ra\la w_k|+|z_k][ w_k|}{\sqrt{\la w_k|w_k\ra\la z_k|z_k\ra}}\,.
\ee
And we check that it actually maps each $w_i$ to the corresponding $z_i$:
\be
\forall i,\quad
g_k\,|w_i\ra
\,=\,
\f{|z_k\ra\la w_k|w_i\ra+|z_k][ w_k|w_i\ra}{\sqrt{\la w_k|w_k\ra\la z_k|z_k\ra}}
\,=\,
\f{|z_k\ra\la z_k|z_i\ra+|z_k][ z_k|z_i\ra}{\la z_k|z_k\ra}
\,=\,
|z_i\ra\,.
\ee

\end{proof}

As a result, the  key point is that the $\SL(2,\C)$ invariant observables $F_{ij}$ entirely determine a unique (framed) polyhedron.

A na\"ive puzzle is that there are $N(N-1)/2$ such observables $F_{ij}$, thus giving $N(N-1)$ real parameters, while the space of framed polyhedra is of dimension $(4N-6)$.
This points to the fact that the $F$'s variables are not independent and satisfy the Pl\"ucker relations (which can be directly checked from their explicit definition in terms of the spinors):
\be
\label{plucker}
\forall i,j,k,l,\qquad
F_{ij}F_{kl}=F_{ik}F_{jl}-F_{il}F_{jk}\,.
\ee
Applying this to $k,l=1,2$, we get\footnotemark:
$$
F_{ij}F_{12}=F_{i1}F_{j2}-F_{i2}F_{j1}.
$$
This means that we can obtain all the $F_{ij}$ from the two $(N-2)$-dimensional complex vectors $F_{i1}$ and $F_{i2}$ (for $i\ge 3$) plus the scale factor $F_{12}$. This minimal data is defined in terms of $2(N-2)+1$ complex parameters thus $(4N-6)$ real parameters as expected.
\footnotetext{
Reversely, if the antisymmetric matrix $F_{ij}$ satisfies the Pl\"ucker relations for $k,l=1,2$ and arbitrary $i,j$, then it satisfies the full Pl\"ucker relations for an arbitrary quadruplet:
\beq
F_{ij}F_{kl}
&=&
\f{1}{F_{12}^2}(F_{i1}F_{j2}-F_{i2}F_{j1})(F_{k1}F_{l2}-F_{k2}F_{l1})\nn\\
&=&
\f{1}{F_{12}^2}
\Big{[}
(F_{i1}F_{k2}-F_{i2}F_{k1})(F_{j1}F_{l2}-F_{j2}F_{l1})
-(F_{i1}F_{l2}-F_{i2}F_{l1})(F_{j1}F_{k2}-F_{j2}F_{k1})
\Big{]}
\,=\,
F_{ik}F_{jl}-F_{il}F_{jk}\,.\nn
\eeq
}

This is illustrated by the fact that one can send by a $\SL(2,\C)$ transformation\footnotemark an arbitrary set of spinors $z_i$ on a new set of spinors such that the first two spinors are collinear to the complex vectors $(1,0)$ and $(0,1)$:
$$
\exists \Lambda\in\SL(2,\C),\qquad
\Lambda\,|z_1\ra=\sqrt{F_{12}}\mat{c}{1\\0},\quad
\Lambda\,|z_2\ra=\sqrt{F_{12}}\mat{c}{0\\1},\quad
\Lambda\,|z_{i\ge 3}\ra=\f1{\sqrt{F_{12}}}\mat{c}{-F_{i2}\\F_{i1}}\,.
$$
\footnotetext{
The $\SL(2,\C)$ used here on the space of spinors $\C^{2N}$ are slightly different than the ones previously used in \cite{closure,tetrahedron} to study the Kapovich-Millson phase space of polyhedra and the resulting Hilbert space of intertwiners. Working with fixed face areas, the phase space is the product of $N$ 2-spheres. Then one can use a (unique) $\SL(2,\C)$ transformation to send the first vector pointing to the north pole, the second to the south pole and the third on the equator along (say) the $x$-axis. The polyhedron is then described by the area $V_i$ of its faces plus the following cross-ratios (giving the direction of the remaining vectors on the unit 2-sphere after having acted with the $\SL(2,\C)$ transformation) best defined in terms of the complex variables $\zeta_i=z_i^1/z_i^0$:
$$
\forall i\ge4,\quad
Z_{i}\,\equiv\,
\f{\zeta_i-\zeta_1}{\zeta_3-\zeta_2}\,.
$$
This parametrizes the polyhedron in terms of the $N$ face areas plus $(N-3)$ complex cross-ratios, which gives the correct dimension, $N+2(N-3)=(3N-6)$.
These cross-ratios can be almost translated in the $F$-variables \cite{generating}, which hints towards an explicit link between the two considered $\SL(2,\C)$ actions:
\be
Z_{i\ge3}
\,=\,
\f{\zeta_i-\zeta_1}{\zeta_3-\zeta_2}
\,=\,
\f{F_{1i}}{F_{23}}
\f{z^0_2z^0_3}{z^0_1z^0_i}
\ee
The $\SL(2,\C)$ action, the fibration of the phase space in terms of $\SL(2,\C)$-orbits and the parametrization in terms of cross-ration turned out powerful when constructing coherent intertwiner states and studying the integration measure over them \cite{closure,tetrahedron,generating}. In particular, it hints towards a link between coherent intertwiner states and conformal field theory. The possible reformulation of our spinor phase space in terms of conformal field theory is postponed to future investigation.
}

\subsection{The Cyclic $\U(N)$ Action}

We now come to the key tool of this paper: $\U(N)$ transformations acting on framed polyhedra with $N$ faces.
Following \cite{un2,un3,spinor}, we introduce the natural action of the $\U(N)$ group on collections of $N$ spinors in $\C^{2N}$:
\be
\{z_i\}_{i=1..N}
\quad\longrightarrow
\{(Uz)_i=\sum_j U_{ij}z_j\}\,.
\ee
The key point is that this action commutes with the closure constraints:
$$
\sum_i |(Uz)_i\ra\la (Uz)_i|
\,=\,
\sum_{i,j,k} \overline{U_{ij}}U_{ik}\, |z_j\ra\la z_k|
\,=\,
\sum_{j,k} (U^\dagger U)_{jk}\, |z_j\ra\la z_k|
\,=\,
\sum_i |z_i\ra\la z_i|\,.
$$
Thus this induces an action of unitary group on the space $\cP_N^z$ of framed polyhedra. Moreover, taking the trace of the previous equation, we check that this action leaves invariant the total area of the polyhedron:
$$
\sum_i \la (Uz)_i|(Uz)_i\ra =\sum_i \la z_i|z_i\ra \,.
$$

Notice that this action does not simply act on the 3-vectors $\vV_i$ but also involves the individual phases of each spinor. Therefore we truly need the spinors and one can not simply define a $\U(N)$-action on the space of polyhedra $\cP_N$.

\medskip

At the infinitesimal level, this action is generated by the scalar products\footnotemark between the spinors \cite{un2,un3,spinor1}:
\be
E_{ij}=\la z_i|z_j\ra\,,
\qquad
\{E_{ij},|z_k\ra\}
\,=\,
i\,\delta_{ik}\,|z_j\ra\,,
\qquad
e^{i\{\sum_{i,j}\alpha_{ij}E_{ij},\cdot\}}\,|z_k\ra
\,=\,
|(e^{i\alpha}\,z)_k\ra\,,
\ee
where $e^{i\alpha}\in\U(N)$ if the matrix $\alpha$ is Hermitian. As we said in the previous section, these generators commute with the closure constraints generating the $\SU(2)$ transformations on the spinors, $\{\vcC,E_{ij}\}=0$,
confirming that $\U(N)$ transformations commute with the $\SU(2)$ action. Finally, we look at their Poisson bracket \eqref{E_un} and check that they form the expected $\u(N)$ Lie algebra.
\footnotetext{
We can also compute the action on the spinors generated by the observables $F_{ij}=[z_i|z_j\ra$. It is more complicated than for the $\U(N)$ transformation since it will mix the spinors with their dual (mixing holomorphic and anti-holomorphic components). A straightforward calculation gives for an anti-symmetric matrix $\beta$:
\beq
&&e^{\f i2\{\sum_{i,j}(\beta_{ij}F_{ij}+\bbeta_{ij}\bF_{ij})\,,\,\cdot\}}\,
|z_k\ra
\,=\,
\Big{(}\delta_{kj}+\f12 (\bbeta \beta)_{kj}+\f1{4!}(\bbeta \beta)^2_{kj}\Big{)}
\,|z_j\ra
+\Big{(}\delta_{ki}+\f1{3!} (\bbeta \beta)_{ki}+\f1{5!}(\bbeta\beta)^2_{ki}\Big{)}
\,\bbeta_{ij}\,|z_j]\,,\nn\\
&&e^{\f i2\{\sum_{i,j}(\beta_{ij}F_{ij}+\bbeta_{ij}\bF_{ij})\,,\,\cdot\}}\,
|z_k]
\,=\,
\Big{(}\delta_{kj}+\f12 (\beta\bbeta )_{kj}+\f1{4!}(\beta\bbeta )^2_{kj}\Big{)}
\,|z_j]
+\Big{(}\delta_{ki}+\f1{3!} (\beta\bbeta )_{ki}+\f1{5!}(\beta\bbeta )^2_{ki}\Big{)}
\,\beta_{ij}\,|z_j\ra\,.\nn
\eeq
Contrarily to the $\U(N)$ transformations, these do not leave invariant the total area of the polyhedron, as one can check directly from the Poisson brackets of the $F$'s and $\bF$'s with $2\lambda=\sum_k \la z_k|z_k\ra$:
\be
\{F_{ij}\,,\,\sum_k \la z_k|z_k\ra\}\,=\,
-2i\,F_{ij}\,,
\qquad
\{\bF_{ij}\,,\,\sum_k \la z_k|z_k\ra\}\,=\,
+2i\,\bF_{ij}\,. \nn
\ee
}

We further check that these generators commute with the total area of the polyhedron $2\lambda=\sum_i \la z_i|z_i\ra$, thus confirming that the total area is invariant under $\U(N)$ transformations.

\medskip

The key feature of this $\U(N)$-action on the space of framed polyhedron is that the action is cyclic. Indeed, we can reach any configuration up to a global scale factor from the completely degenerate and flat configuration by an arbitrary $\U(N)$ transformation. More precisely, we introduce the trivial reference point:
\be
|\Om_1\ra=\mat{c}{1\\0},
\quad
|\Om_2\ra=|\Om_1]=\mat{c}{0\\1},
\quad
|\Om_3\ra=..=|\Om_N\ra=0\,,
\ee
which obviously satisfies the closure constraints, $\sum_i |\Om_i\ra\la \Om_i|= \id$. The corresponding 3-vectors are the unit vector in the $z$-direction, $\vV_1=\hat{e}_z$, its opposite $\vV_2=-\vV_1=-\hat{e}_z$, and vanishing vectors $\vV_3=..\vV_N=0$, thus giving a completely-flat configuration defining a degenerate polyhedron.
Acting with an arbitrary $\U(N)$ transformation on this configuration gives:
$$
|(U\Om)_k\ra= \mat{c}{U_{k1}\\U_{k2}}\,.
$$
Reversely, considering from an arbitrary collection of $N$ spinors $\{z_i,\,i=1..N\}$ satisfying the closure constraints, we can rescale it so that it is of the form above:
\be
\label{z_U}
|z_k\ra
\,=\,
\sqrt{\lambda}\,\mat{c}{U_{k1}\\U_{k2}}
\,=\,
\sqrt{\lambda}\,|(U\Om)_k\ra\,,
\qquad
\lambda=\f12\sum_i \la z_i|z_i\ra\,.
\ee
This works because the closure constraints are equivalent to the fact that the first and second components of the spinors form two orthogonal complex $N$-vectors with equal norms:
$$
\vcC=0\quad\Longleftrightarrow\qquad
\sum_i \bz_i^0 z_i^1=0
\quad\textrm{and}\quad
\sum_i |z^0_i|^2=\sum_i |z^1_i|^2=\lambda\,,
$$
exactly the same as the first two columns, $(\sqrt{\lambda}\,U_{k1})_k$ and $(\sqrt{\lambda}\,U_{k2})$, of a unitary matrix $U\in\U(N)$ re-scaled by $\sqrt{\lambda}$.

\medskip

Moreover, the stabilizer group of the completely flat polyhedron clearly is $\U(N-2)$. Thus the set of collections of $N$ spinors satisfying the closure constraint is identified to the quotient $\U(N)/\U(N-2)$. Further quotienting by the action of $\SU(2)$ (to get equivalence classes of poyhedra under 3d rotations), this lead us to the following proposition as hinted in \cite{un1,un2}:

\begin{prop}
We have an action of the unitary group $\U(N)$ on the space $\cP^z_N$ of framed polyhedron with $N$ faces. This leads to an isomorphism between  $\cP^z_N=\C^{2N}//\SU(2)$ and the Grassmannian space $\U(N)\,/\,(\SU(2)\times\U(N-2))$. In particular, we have the equivalence for a set of spinors $z_i\in\C^{2N}$:
\be
\sum_i |z_i\ra\la z_i| \propto\id_2
\qquad\Longleftrightarrow\qquad
\exists\lambda\in\R_+,\,\,\exists U\in\U(N),\,\,
\forall i\,,\quad|z_i\ra=\sqrt{\lambda}\,|(U\Om)_i\ra
\ee
\end{prop}

\medskip

Before moving on to the next part of the paper, we would like to re-visit this $\U(N)$ structure of the space of polyhedra from the point of view of the $\SU(2)$-invariant observables. The definition of the spinors $|z_k\ra=\sqrt{\lambda}\,|(U\Om)_k\ra$ in terms of the unitary matrix $U\in\U(N)$ implies the diagonalization of the observables $E_{ij}$ and $F_{ij}$ as $N\times N$ matrices:
\be
|z_k\ra=\sqrt{\lambda}\,|(U\Om)_k\ra
\quad\Longrightarrow\qquad
E=\lambda\,\bU\,\mat{cc|c}{1&0&\\0&1&\\ \hline&&0_{N-2}}\,{}^tU\,,
\quad
F=\lambda\,U\,\mat{cc|c}{0&1&\\-1&0&\\ \hline&&0_{N-2}}\,{}^tU\,,
\ee
where the off-diagonal components vanish and ${}^tU=\bU^{-1}$. These definitions of $E$ and $F$ in terms of the matrix $U$ are invariant under transformations $U\arr UG$ with $G\in \SU(2)\times \U(N-2)$.
We can also deduce the existence of the unitary matrix $U$ directly from the $E$'s or $F$'s. Indeed, first considering the Hermitian matrix of the scalar products $E_{ij}=\la z_i|z_j\ra$ for a closed configuration of spinors $|z_i\ra$, the matrix $E$ satisfies a simple polynomial identity:
$$
(E^2)_{ij}=\sum_{k} \la z_i|z_k\ra\la z_k|z_j\ra
=\lambda \sum_{k} \la z_i|z_j\ra
=\lambda E_{ij}\,,
\qquad\textrm{with}\quad
\lambda=\f12\sum_k \la z_k|z_k\ra=\f{\tr E}2\,.
$$
Reversely, this equality is obviously enough to guarantee the existence of $U$ (as already stated in \cite{spinor1}):
\begin{res}
Considering a $N\times N$  Hermitian matrix $E$ satisfying $E^2=\f{\tr E}2\,E$ for some $\lambda\in\R_+^*$, it is diagonalizable with $\lambda$ as its single non-vanishing and doubly-degenerate eigenvalue:
$$
\exists U\in\U(N),\quad
E= \lambda\,\bU\,\mat{c|c}{\id_2&\\
\hline&0_{N-2}}\,{}^tU\,,
\qquad\textrm{with}\quad
2\lambda={\tr E}\,.
$$
\end{res}

One can also start from the matrix of  observables $F_{ij}$:
\begin{res}
Considering a non-vanishing $N\times N$  matrix $F$ satisfying the Pl\"ucker relations \eqref{plucker}, it is automatically antisymmetric and of the following form:
$$
\exists U\in\U(N),\quad
\exists \lambda\in\R_+,\quad
F=\lambda\,U\,\mat{cc|c}{0&1&\\-1&0&\\ \hline&&0_{N-2}}\,{}^tU\,.
$$
\end{res}
\begin{proof}
We specialize the Pl\"ucker relations to a doublet of indices $i,j$ and the fixed indices $k,l=1,2$ as before:
$$
F_{ij}F_{12}=F_{i1}F_{j2}-F_{i2}F_{j1}.
$$
This means that $F$ is antisymmetric and furthermore that it is of rank 2 (if it is non-vanishing).

$F$ being a complex antisymmetric matrix, one can diagonalize it as $F=U\,\Sigma\,{}^tU$, with $U\in\U(N)$ and $\Sigma$ of the following type:
$$
\Sigma\,=\,
\mat{cc|cc|c|ccc}{0&\lambda_1 &&&&&&\\ -\lambda_1&0 &&&&&& \\ \hline
&&0&\lambda_2 &&&&\\ &&-\lambda_2&0 &&&&\\ \hline
&&&&\ddots&&& \\ \hline
&&&&&0&& \\
&&&&&&\ddots &\\
&&&&&&&0}
$$
The $\lambda_k$'s are a priori complex. Since $F$ is of rank 2, there is a single non-vanishing block with $\lambda_1\in\C$. One can then absorb its phase in the definition of the unitary matrix $U$ and keep its modulus as $\lambda\in\R_+$.

\end{proof}

\medskip

In the next part \ref{integrals}, we will use this reformulation of (framed) polyhedra in terms of unitary matrices to compute systematically the averages and correlations between the normal vectors defining  polyhedra and characterizing their shape.

\section{Computing Averages through Integrals on $\U(N)$}
\label{integrals}

Now considering the ensemble of (framed) polyhedra provided with the uniform measure or equivalently the $\U(N)$ Haar measure, we study the averages and correlations of polynomial observables and aim at characterizing the shape of a typical polyhedron. In particular, we show how to formulate the averages of polynomial observables in the normal vector $\vV_i$ as integrals over the unitary group $\U(N)$ and how to use the Itzykson-Zuber integral as a generating function for these.

\subsection{Counting Polyhedra: Entropy}

We start by computing the volume of the space of polyhedra with $N$ faces for a fixed total area.  This corresponds to computing the entropy for a simplified model of the black hole horizon in loop quantum gravity \cite{eugenio_bh,counting}. When quantized, this model reproduces the loop gravity's entropy calculation through counting the dimensions of $\SU(2)$ intertwiner spaces (see \cite{lqg_bh_su2} and \cite{lqg_bh_review} for reviews and detail on the description of the quantum states of a black hole horizon as $\SU(2)$ intertwiners).

One defines the density of framed polyhedra with $N$ faces and fixed area $2\lambda$ as the following straightforward integral over spinor variables constrained by a total area condition and the closure conditions:
\beq
\label{rhoN}
\rho_N[\lambda]
&\equiv&
\,8\pi\,\int \prod_i^N \f{d^4z_i}{\pi^2}\,
\delta\left(\sum_k^N \la z_k|z_k\ra-2\lambda\right)\,
\delta^{(3)}\left(\sum_k^N \la z_k|\vsigma|z_k\ra\right)\\
&=&
\lambda^{2N-4}\,8\pi\,\int \prod_i^N \f{d^4z_i}{\pi^2}\,
\delta\left(\sum_k^N \la z_k|z_k\ra-2\right)\,
\delta^{(3)}\left(\sum_k^N \la z_k|\vsigma|z_k\ra\right)\,,\nn
\eeq
where the $8\pi$-factor is an arbitrary choice of normalization.
Integrating over the phases of the spinors, one can perform the change of variables from the $z_k\in\C^2$ to the vectors $\vV_k\in\R^3$. The change of measure is straightforward to perform \cite{un3,un4} and one obtain the density of polyhedra as previously defined in \cite{counting}:
\beq
\rho_N[\lambda]
&=&
8\pi\,\int \prod_i^N \f{d^3\vV_i}{4\pi V_i}\,
\delta\left(\sum_k^N V_k-2\lambda\right)\,
\delta^{(3)}\left(\sum_k^N \vV_k\right)\\
&=&\lambda^{2N-4}\,8\pi\,\int \prod_i^N \f{d^3\vV_i}{4\pi V_i}\,
\delta\left(\sum_k^N V_k-2\right)\,
\delta^{(3)}\left(\sum_k^N \vV_k\right)\,. \nn
\eeq
The most direct way to compute this integral is to Fourier-transform the $\delta$-distribution\footnotemark. One then gets:
\be
\label{rho0}
\rho_N[\lambda]
\,=\,
\f{\lambda^{2N-4}}{(N-1)!(N-2)!}\,,
\ee
where the $8\pi$-factor had been chosen so that $\rho_2[\lambda]=1$ for polyhedra with $N=2$ faces.
\footnotetext{
The only trick is to introduce a regulator $\eps>0$ before performing the integrals over the normal vectors $\vV_k$ through the Fourier transform identity:
$$
\delta\left(\sum_k^N V_k-2\lambda\right)
\,=\,
e^{+2\eps\lambda}\,\int_\R \f{dq}{2\pi}\,
e^{-2iq\lambda}\,\prod_k^N e^{+iqV_k }e^{-\eps V_k}\,,
$$
valid for arbitrary values of $\eps\in\R^+$.
}
One can find the details of this calculation in appendix \ref{rho}. The method is actually useful for defining a partition function over the ensemble of polyhedra and computing the averages of polynomial observables by differentiation as outlined in \cite{counting}. Another method also shown in \cite{counting} is to Fourier-transform the $\delta$-distribution while keeping the spinor variables. One then gets Gaussian integrals which can be easily handled. We will not use this method here.

\medskip

Instead, we would like to highlight the fact that the space of framed polyhedra is isomorphic to the Grassmaniann space $\U(N)/\U(N-2)\times\SU(2)$, which allows for a more geometric interpretation for the volume of $\cP^z_N$.

Indeed, keeping the spinor variables in the definition \eqref{rhoN} of the density $\rho_N[\lambda]$, we write explicitly the total fixed area and closure constraints  in terms of the real and imaginary parts of the spinor variables, $z_k^A=x_k^A+iy_k^A$:
\beq
&&\sum_k |z_k^0|^2=\sum_k |z_k^1|^2=1,
\qquad
\sum_k \bz_k^0 z_k^1 =0, \nn\\
&&\textrm{or equivalently}\quad
\sum_k (x_k^0)^2+(y_k^0)^2=\sum_k (x_k^1)^2+(y_k^1)^2=1,
\qquad
\sum_k x_k^0x_k^1+y_k^0y_k^1=
\sum_k x_k^0y_k^1-y_k^0x_k^1=
0\,.
\eeq
This means that we have two unit vectors of dimension $2N$, $(x_k^0,y_k^0)$ and $(x_k^1,y_k^1)$, both on the $(2N-1)$-dimensional sphere $\cS_{2N-1}$. The second vector $(x_k^1,y_k^1)$ is actually orthogonal to the first vector $(x_k^0,y_k^0)$ but also to the vector $(y_k^0,-x_k^0)$ itself orthogonal to the former vector. This means that this second vector $(x_k^1,y_k^1)$ actually lives on a $(2N-3)$-dimensional sphere $\cS_{2N-3}$ still with unit radius. This leads to a simple geometric interpretation of the density of polyhedra with $N$ faces and fixed total area as the product of the volumes of the spheres $\cS_{2N-1}$ and $\cS_{2N-3}$:
\be
\rho_N[\lambda]
\,=\,
\lambda^{2N-4}\,\f{\pi}4 \f1{(\pi^2)^N}\,\textrm{Vol}(\cS_{2N-1})\,\textrm{Vol}(\cS_{2N-3})
\,=\,
\lambda^{2N-4}\,\f{\pi}4 \f1{(\pi^2)^N}\,\f{2\pi^N}{(N-1)!}\,\f{2\pi^{N-1}}{(N-2)!}
\,=\,
\f{\lambda^{2N-4}}{(N-1)!(N-2)!}\,,
\ee
where the factor $\pi/4$ adjusts the over-all normalization of the integrals.

\medskip

From the point of view of unitary groups, the situation is clear: we are computing the volume of the coset $\U(N)/\U(N-2)$, which can be decomposed as $\U(N)/\U(N-1)\,\times\,\U(N-1)/\U(N-2)$, which is isomorphic to the product of the two spheres $\cS_{2N-1}\,\times\,\cS_{2N-3}$.

Below, we will analyze the average of polynomial observables over the ensemble of polyhedra and we will fully use for this purpose the $\U(N)$ structure. In practice, we will normalize all the results by the overall volume of the space of polyhedra at fixed total area by simply using the normalized Haar measure on the unitary group $\U(N)$.

\subsection{Probing the Average Geometry of a Polyhedron and Fluctuations}
\label{integrals1}

We would like to characterize a typical polyhedron drawn at random from the ensemble with the Haar measure on $\U(N)$. To this purpose, we compute the averages of the normal vectors and  their correlations.
Using the explicit expression of the spinors and vectors in terms of the unitary matrix $U\in\U(N)$ as given earlier by \eqref{z_U},
\be
|z_k\ra\,=\, \sqrt{\lambda}\,\mat{c}{U_{k1}\\U_{k2}},
\qquad
V_k=\la z_k|z_k\ra \,=\,\lambda \,\sum_{\alpha=1,2}\bU_{k\alpha}U_{k\alpha},
\quad
V_k^a=\la z_k|\sigma^a|z_k\ra
\,=\,\lambda \,\sum_{\alpha,\beta}\bU_{k\alpha}U_{k\beta} \sigma^a_{\alpha\beta},
\ee
the averages of product of the norms $V_k$ or vector components $V_k^a$ can all be re-cast as polynomial integrals over $\U(N)$ of the type:
\be
\int_{\U(N)} dU\, U_{i_1j_1}U_{i_2j_2}..U_{i_nj_n}\bU_{k_1l_1}..\bU_{k_nl_n}\,,
\ee
where the number of $U$'s and of its complex conjugate $\bU$'s must match else the integral vanishes.
Here we focus on the explicit computation of these integrals up to the 4rth order, using the basic recoupling theory of $\U(N)$ representations, in order to probe the average geometry and uncertainty of the polyhedra.
Below, we will give the generic behavior of the polynomial integrals in section \ref{correlation} and discuss how such integrals can be generated from the Itzykson-Zuber formula in section \ref{IZ}.

\medskip

Starting with quadratic integrals, we compute the average norm of each normal vector:
\be
\label{Vcl}
\la V_k \ra
\,=\,
\lambda\,\int dU\,(\bU_{k1}U_{k1}+\bU_{k2}U_{k2})
\,=\,
\f{2\lambda}N\,,
\ee
using the orthogonality of the matrix elements of a $\U(N)$ group element in the fundamental $N$-dimensional representation.
This was expected since the total area is $2\lambda$, which is shared isotropically among the $N$ normal vectors. Beside this, the average of each of the vector components $\la V_k^a \ra$ vanishes.

The next step is to compute  the quartic integrals $\la V_k^2\ra$ and $\la V_k^a V_l^b\ra$. This is done using the explicit formula (computed by decomposing the tensor product $U\otimes \bU$ as the matrix elements of the group element $U$ in the trivial and adjoint representations):
\beq
&&\int_{\U(N)} dU\, U_{ij}\bU_{\alpha\beta}U_{\mu\nu}\bU_{kl} \\
&&\,=\,
\f1{N^2}\delta_{i\alpha}\delta_{k\mu}\delta_{j\beta}\delta_{l\nu}
+\f1{N^2-1}
\left(
\delta_{ik}\delta_{\alpha\mu}\delta_{jl}\delta_{\beta\nu}
-\f1N \delta_{ik}\delta_{\alpha\mu}\delta_{j\beta}\delta_{l\nu}
-\f1N\delta_{i\alpha}\delta_{k\mu}\delta_{jl}\delta_{\beta\nu}
+\f1{N^2}\delta_{i\alpha}\delta_{k\mu}\delta_{j\beta}\delta_{l\nu}
\right)\,.\nn
\eeq
Applying this to the average squared-norm and correlations between vector components, straightforward calculations give:
\be
\label{V2cl}
\la V_i^2\ra
\,=\, \lambda^2 \,\int dU\, U_{i\alpha}\bU_{i\alpha}U_{k\nu}\bU_{k\nu}
\,=\,
\f{6\lambda^2}{N(N+1)}\,,
\qquad
\la V_i^a V_i^b\ra
\,=\,
\f{+2\lambda^2\delta_{ab}}{N(N+1)}\,,
\ee
\be
\la V_iV_j\ra_{i\ne j}
= \lambda^2\,\f{2(2N-1)}{(N-1)N(N+1)}\,,
\qquad
\la V_i^a V_j^b\ra
\,=\,
\f{-2\lambda^2\delta_{ab}}{N(N^2-1)}\,.
\ee
First, this allows to compute the spread of a face area:
\be
\sqrt{\la V_i^2\ra-\la V_i\ra^2}
\,=\,\f{\lambda\sqrt{2}}{N}\,\sqrt{\f{N-2}{N+1}}
\quad\underset{N\gg1}\sim\,\f{\la V_i\ra}{\sqrt{2}}\,,
\ee
which means that the probability distribution of the area of a face remains fuzzy even as the number of faces grows.

Second, looking at the correlation $\la V_iV_j\ra$ between the areas of two distinct faces, we can check that $\sum_{i,j}\la V_iV_j\ra=4\lambda^2$ as expected from the fixed total area constraint $\sum_i V_i=2\lambda$. Moreover, we check that the area of faces becomes more and more decoupled as the number of faces grows:
\be
\f{\la V_i\ra\la V_j\ra-\la V_iV_j\ra}{\la V_i\ra\la V_j\ra}
\,=\,
\f{N-2}{2(N^2-1)}
\quad
\underset{N\arr\infty}\longrightarrow\,0\,.
\ee

Third, we introduce another  set of observables $\Theta_{ab}$ characterizing the shape of a polyhedron:
\be
\Theta_{ab}=\sum_i V_i^aV_i^b -\f13\delta^{ab}V_iV_i\,,
\ee
which vanishes if the normal vectors are distributed spherically, but will be non-vanishing as soon as we deviate from the isotropic distribution (e.g. if the shape of the polyhedron is more ellipsoidal than spherical). Here, we easily check that:
\be
\la \Theta_{ab} \ra =0\,.
\ee

\medskip

Instead of using $\U(N)$-integrals, one could instead compute brutally these averages and correlations as integrals over the normal vectors together with the closure constraints and fixed area constraint. We give the explicit method in appendix \ref{brutal} and we recover the formulas above. But we have further computed the mean value $\la \Theta_{ab}\Theta_{cd} \ra$  in order to get the standard deviation from the spherical configuration:
\beq
\la \Theta^{ab}\Theta^{cd} \ra
&=&
\lambda^4\,
\f{4\left(4(N^2+N-2)\delta^{ab}\delta^{cd}-6(N-1)(\delta^{ac}\delta^{bd}+\delta^{ad}\delta^{bc})\right)}{3(N-1)N(N+1)(N+2)(N+3)}\\
&\sim& N^{-3}\quad\underset{N\arr\infty}{\longrightarrow}\,0\,.
\eeq
This means that the probability distribution over the ensemble of polyhedra is highly peaked about the spherical configuration.
To get a simpler single indicator, we can compute the average of $\tr \Theta^2$. Classically, $\tr\Theta^2$ has direct expression in terms of the vector scalar products:
\be
\tr\Theta^2
\,=\,
\sum_{i,j}(\vV_i\cdot\vV_j) ^2 -\f13\left(\sum_i V_i^2\right)^2.
\ee
It is always null or positive, $\tr \Theta^2 \ge0$, and measures somehow the shape of the polyhedron. It is maximal when the polyhedron is flat and gets smaller as the polyhedron becomes more and more spherical. For instance, it vanishes for a cube ($N=6$ faces).
We get its average by contracting the indices in the formula above:
\be
\la \tr\Theta^2 \ra
\,=\,
4\lambda^2\f{(N-4)}{N(N+1)(N+2)(N+3)}
\,\sim\,
N^{-3}\quad\underset{N\arr\infty}{\longrightarrow}\,0\,.
\ee

This can be compared to the {\it concentration of measure} on the sphere $\cS_{2N-1}\,\sim\,\U(N)/\U(N-1)$  induced by the Haar measure on $\U(N)$: the uniform measure  concentrates very strongly about any equator as $N$ grows large (see e.g. \cite{hayden} for a description of this phenomenon, focusing on its application to the entanglement of random states). We very probably have a similar concentration of measure on the coset $\U(N)/\U(N-2)$. We will have a closer look at this later in section \ref{explicit}.

\medskip

It is interesting to compare these averages to the one of an ensemble of normal vectors without the closure constraints:
\be
\rho^0_N[\lambda]
\,\equiv\,
\int \prod_i^N \f{d^3\vV_i}{4\pi V_i}\,
\delta\left(\sum_k^N V_k-2\lambda\right)\,.
\ee
We use the similar brute-force method by Fourier-transforming the $\delta$-distribution, as done in \cite{counting}, with $\eps\in\R_+$:
\be
\rho^0_N[\lambda]
\,=\,
\int \prod_i^N \f{d^3\vV_i}{4\pi V_i}
\int\f{dq}{2\pi}\,
e^{(iq-\eps)(\sum_k^N V_k-2\lambda)}
\,=\,
e^{2\eps\lambda}\int\f{dq}{2\pi}\,
e^{-2iq\lambda}\,I^0(q)^N,
\ee
$$
\textrm{with}\qquad
I^0(q)
=\int\f{d^3\vV}{4\pi V}e^{-\eps V}e^{iqV}
=\int_0^{+\infty}dV\,Ve^{-\eps V}e^{iqV}
=\f1{(\eps-iq)^2}\,.
$$
This allows us to compute this volume:
\be
\rho^0_N[\lambda]
\,=\,
e^{2\eps\lambda}\int\f{dq}{2\pi}\,
e^{-2iq\lambda}\,\f1{(\eps-iq)^{2N}}
\,=\,
\f{\lambda^{2N-1}}{(2N-1)!}
\quad>\,
\rho_N[\lambda]\,.
\ee
Thinking in terms of spinors, this correspond to the (properly normalized) volume of a $(4N-1)$-dimensional sphere. Using the same techniques as given in appendix \ref{brutal} of differentiating with respect to the momentum conjugated to the vectors $\vV_k$, we have computed the averages and correlations of the vector components, which we note with the subscript $(0)$ to distinguish them from the average over the space of polyhedra:
\be
\la V_i\ra^{(0)}\,=\,\f{2\lambda}{N},\quad
\la V_i^2\ra^{(0)}\,=\,\f{3(2\lambda)^2}{N(2N+1)},\quad
\la V_i^aV_i^b\ra^{(0)}\,=\,\f{\delta^{ab}(2\lambda)^2}{N(2N+1)}\,,
\ee
\be
\la V_iV_j\ra^{(0)}_{i\ne j}
\,=\,
\f{2(2\lambda)^2}{N(2N+1)},\qquad
\la V_i^aV_j^b\ra^{(0)}\,=\,0\,.
\ee
At leading order in $N$, we find the same average $\la V_i\ra$ and spread $\la V_i^2\ra$ for the individual face areas. Here, we can easily go further and compute exactly all the averages $\la V_i^n\ra$ for an individual face area. Indeed:
$$
\la V_i^n\ra^{(0)}
\,=\,
\f1{\rho^0_N[\lambda]}\,
e^{2\eps\lambda}\int\f{dq}{2\pi}\,
e^{-2iq\lambda}\,I^0(q)^{N-1}I^n(q)\,,
$$
$$
\textrm{with}\qquad
I^n(q)
\,=\,
\int_0^{+\infty}{dV}\,V^{n+1}e^{-\eps V}e^{iqV}
=(-\pp_\eps)^n\,I^0(q)
=(n+1)!\,(\eps-iq)^{-(n+2)}\,,
$$
which gives:
\be
\label{Vn0}
\la V_i^n\ra^{(0)}
\,=\,
(2\lambda)^n\,\f{(n+1)!(2N-1)!}{(2N+n-1)!}\,.
\ee
Furthermore, the closure condition is obviously satisfied in average $\la\sum_i V_i^a\ra^{(0)}\,=0$, but it now has a on-trivial spread:
\be
\la|\sum_i \vV_i|^2\ra^{(0)}
\,=\,
\f{3(2\lambda)^2}{2N+1}\,.
\ee
This is due to the vanishing of the correlation between components of two distinct vectors $i$ and $j$. Indeed the main difference between the ensembles satisfying or not the closure constraints is in the correlations between normal vectors. For an individual vector, it does not change the leading order (in $N$) of the averages of the powers of the area $\la V_i^n\ra$ (though the exact full expression does change), as we will check later in section \ref{IZ}.

Going further, we easily check that $\la \Theta^{ab} \ra^{(0)}=0$ and that the ensemble is also peaked on spherically symmetric sets of vectors. We nevertheless expect a deviation for the averages $\la \Theta^{ab}\Theta^{cd} \ra^{(0)}$ but we haven't checked this explicitly.

\medskip

Up to now we have looked explicitly at integrals up to order 4 in the normal vectors (up to order 8 in the spinors). Using the $\U(N)$ framework, it is possible to compute generic formulas for all polynomial integrals over the unitary group and thus compute at least at leading order all polynomial averages over the ensemble of (framed) polyhedra, as we will see in the next section. This is much more powerful than the method of differentiating the partition function.


\subsection{Polynomial Averages at Leading Order}
\label{correlation}

Using the interplay between the irreducible representations  of $\U(N)$ and of the permutation group $S_n$, \cite{colllins} give a systematic formula for polynomial integrals over $\U(N)$:
\be
\label{collins}
\int dU\,U_{i_1j_1}..U_{i_nj_n}\bU_{k_1l_1}..\bU_{k_nl_n}
\,=\,
\sum_{\sigma,\tau\in S_n}
\delta_{i_1k_{\sigma(1)}}..\delta_{j_1l_{\tau(1)}}
\,
\wg^{(n)}_N(\sigma\tau^{-1}))\,,
\ee
where the sum is over permutations $\sigma$ and $\tau$. The factor is given explicitly as
\be
\wg^{(n)}_N(\sigma)
\,\equiv\,
\f1{n!^2}\sum_{\Lambda\vdash n}
\f{\chi^\Lambda(\id)^2\chi^\Lambda(\sigma)}{s_{\Lambda,N}(1)}\,,
\ee
where the sum is over partitions $\Lambda\vdash n$ of the integer $n$, $\chi^\Lambda$ is the corresponding character of the permutation group $S_n$, and $s_{\Lambda,N}(x_1,..,x_N)$ is the corresponding Schur function, with in particular $s_{\Lambda,N}(1)=s_{\Lambda,N}(1,..,1)$ the dimension of the irreducible representation of $\U(N)$ associated with $\Lambda$.

Furthermore, \cite{collins} goes further and uses combinatorics to provide an asymptotic formula for the symbol $\wg$ at large $N$:
\be
\wg^{(n)}_N(\sigma)
\quad\underset{N\arr\infty}{\sim}\quad
\f1{N^K}\prod_{k=1}^K(-1)^{|c_k|}\cC_{|c_k|}\,,
\ee
in terms of the cycle decomposition of the permutation $\sigma=c_1\dots c_K$. For a generic permutation  $|\sigma|$ is the minimal number of transpositions needed to write $\sigma$. For a cycle, $|c|$ is simply the length of the cycle minus one. $\cC_c$ is the $c$-th Catalan number:
$$
\cC_c\,\equiv\,
\f{1}{c+1}\bin{c}{2c}
\,=\,
\f{(2c)!}{c!(c+1)!}
\,\underset{c\gg 1}{\sim}\,\f1{\sqrt{\pi}}\f{2^{2c}}{c^{3/2}}
\,,
$$
in terms of binomial coefficients.
Large $N$ corresponds geometrically to a very large number of faces and thus at a refinement limit for the polyhedra. This will likely be very useful to understand the large $N$ limit of the distribution of polyhedra and thus study their continuous limit.
This result was used in \cite{collins} to study the large $N$ limit of the Itzykson-Zuber formula, or more precisely of its derivative,
$$
\lim_{N\arr\infty}\,
\left.\f{\pp^n}{\pp\theta^n} \f1{N^2}\,
\log\int_{\U(N)}dU\,e^{\theta N\tr(XUYU^\dagger)}\,\right|_{\theta=0}\,
$$
when the normalized traces $N^{-1}\tr X^k$ and $N^{-1}\tr Y^k$ converge at large $N$ (for all $k$'s). We will investigate below how the Itzykson-Zuber formula can actually be used as the generating function for these polynomial integrals over $\U(N)$.

Applying this formula to the product of vector components $V^a_i$, the indices $j$'s and $l$'s in the integral \eqref{collins} will all be equal to 1 or 2 and contracted with Pauli matrices $\sigma^a_{lj}$. The indices $i$'s and $k$'s correspond to the index of the vectors between 1 and $N$. The permutation $\sigma$ have to match the $i$'s with the $k$'s, thus does not mix between different vectors, while the permutation $\tau$  have to match the $j$'s with the $l$'s and can mix terms corresponding to different vectors. Then one has to compute the traces of product of Pauli matrices corresponding to the cycles of the permutation $\tau$.

Thus in theory it is possible to compute systematically the average of any polynomial observables using this formula. In practice, this can become tedious. Nevertheless, the structure of the formula is rather simple (in terms of the permutations $\sigma$ and $\tau$) and one could study in a straightforward manner the averages of the powers of an interesting observable (e.g. the individual face area or the volume of the polyhedron) if one wanted.

\medskip

An equivalent formula but worded differently can be found  in \cite{brandao}, related to the evaluation of the twirling operator in quantum information and used in the context of the convergence to equilibrium under a random Hamiltonian. Considering the Hilbert space of $\otimes_{k=1}^n \cC^N$, on which the unitary operators $U^{\otimes n}$ act. We consider the representation of the permutation group $S_n$ defined by swapping subsystems:
\be
\forall \sigma\in S_n,\qquad
D^{(N)}(\sigma)(e_{i_1}\otimes..\otimes e_{i_n})
\,=\,
e_{\sigma^{-1}(i_1)}\otimes..\otimes e_{\sigma^{-1}(i_n)}\,,
\ee
where $\{e_i\}_{i=1..N}$ forms a basis of $\cC^N$. Following the notations of \cite{brandao}, we write $\cV_\sigma$ as short for the operator $D^{(N)}(\sigma)$.
It is easy to compute the character of this representation $D^{(N)}$:
$$
\chi^{D^{(N)}}(\sigma)
\,=\,
N^{\ell(\sigma)}\,,
$$
where $\ell(\sigma)$ is the number of cycles in the cycle decomposition of the permutation $\sigma$.

Then defining the twirling operator $T_n(\cdot)=\int dU\,U^{\otimes n}(\cdot)U^{\otimes n}{}^\dagger$, we have for any two operators $A,B$ acting on $(\cC^N)^{\otimes n}$:
\be
\tr\,AT_n(B)
\,=\,
\sum_{\sigma,\tau}
a_\sigma \,b_\tau\, M^{-1}_{\sigma\tau}
\qquad\textrm{with the matrix}\quad
M_{\sigma\tau}=\tr \cV_{\sigma^{-1}}\cV_\tau\,,
\ee
and the vectors $a_\sigma\,\equiv\,\tr\,A\cV\cV_{\sigma^{-1}}$ and the same for $B$. The proof can be found in \cite{brandao}. The matrix $M$ has a simple form:
$$
M_{\sigma\tau}
=\tr \cV_{\sigma^{-1}\tau}
=\chi^{D^{(N)}}(\sigma^{-1}\tau)
=N^{\ell(\sigma^{-1}\tau)}\,,
$$
and the whole issue is to invert this matrix, which leads to the same result as presented above when applied to operators $A$ and $B$ taken in the standard basis. More details on the structure and possible computation of $M^{-1}$ can be found in \cite{brandao} for the interested reader.

\subsection{Itzykson-Zuber Formula as Generating Function}
\label{IZ}

The Itzykson-Zuber formula allows to compute the integral over $\U(N)$ of the exponential of matrix elements of $U$ and $\bU$. Based on the localization of integrals, it first appeared in relation to matrix models and two-dimensional quantum gravity \cite{iz1} and be computed explicitly using the Harish-Chandra formula (e.g. \cite{collins}).

It goes as follows. Let us consider two N$\times$N matrices $X$ and $Y$ and let $(x_i)$ and $(y_i)$ be their respective eigenvalues. We call $\Delta(X)=\prod_{i<j} (x_j-x_i)$ and $\Delta(Y)=\prod_{i<j} (y_j-y_i)$ their Vandermonde determinant. Then the Itzykson-Zuber formula reads:
\be
\int_{\U(N)}dU\,
e^{i\theta\,\tr \left(YU^\dagger X U\right)}
\,=\,
\f{\det\left(e^{i\theta x_jy_k}\right)_{1\le j,k\le N}}{\Delta(X)\Delta(Y)}\,(i\theta)^{\f{-N(N-1)}2}\,.
\ee

Choosing appropriate matrices $X$ and $Y$, this Itzykson-Zuber formula can be seen as the generating function for all the correlations between the normal vectors over our polyhedron ensemble.
In our case, let us give an example with the observable $V_i$ and its powers. We have:
$$
V_i
=\la z_i|z_i\ra
=\lambda\,(U_{i1}\bU_{i1}+U_{i2}\bU_{i2})
\,=\, \tr\,(YU^\dagger XU),
\qquad\textrm{with}\quad
Y_{jk}=(\delta_{j1}\delta_{k1}+\delta_{j2}\delta_{k2})
\quad\textrm{and}\quad
X^{(i)}_{jk}=\delta_{ji}\delta_{ki}\,.
$$
The matrix $Y$ is fixed and implements the reduction from $\U(N)$ to our space of polyhedron  $\U(N)/\U(N-2)$. The matrix $X$ selects the considered observables.
Then the mean value $\la \exp(i\theta V_i)\ra$ is a Itzykson-Zuber integral:
\beq
\la e^{i\theta V_i}\ra
&=&
c\int_{\U(N)}dU\,e^{i\theta\,\tr \left(YU^\dagger X U\right)} \nn\\
&=&
1+\sum_{n=1}^\infty \f{(N-1)!}{(n+N-1)!}{(n+1)}\,(i\theta\,\lambda)^n\,,
\eeq
where $c$ is a normalization constant such that $\la 1\ra =1$ for $\theta=0$. The trick to derive this formula is to regularize the Itzykson-Zuber formula by shifting slightly all the eigenvalues of $X$ and $Y$ to ensure that they are different and then to send these regulators to 0 at the end. Then this result gives us directly all the mean values $\la( V_i)^n\ra$, without having to suitably differentiate the density of state $\rho_N[\lambda]$ as in section \ref{integrals1} or compute the polynomial $\U(N)$ integrals as in section \ref{correlation}:
\be
\la V^n\ra
\,=\,
\lambda^n\,
\f{(n+1)!(N-1)!}{(N+n-1)!}\,,
\ee
which matches our expressions already derived for $\la V\ra$ and $\la V^2\ra$. We can compare them to the free model without closure constraints as introduced earlier in section \ref{integrals1}, which had the following averages \eqref{Vn0}:
$$
\la V^n\ra^{(0)}
\,=\,
(2\lambda)^n\,\f{(n+1)!(2N-1)!}{(2N+n-1)!}\,.
$$
First, we notice that these are different (though similar), showing that the two models are clearly distinct and have a different probability distribution for the individual face areas. Second, as claimed earlier, the two expressions nevertheless match at large $N$ for a fixed power $n$:
$$
\la V^n\ra
\underset{N\gg1}{\sim}
\lambda^n\,\f{(n+1)!}{N^n}
\sim
\la V^n\ra^{(0)}\,.
$$

We can go further and get the formula for the fixed matrix $Y$ but for arbitrary matrix $X$. We perturb around the actual eigenvalues of $Y$ as $y_1=1+\eps_1$, $y_2=1+\eps_2$ and $y_{k\ge 3}=\eps_k$. Both numerator and denominator of the Itzykson-Zuber vanish as all the $\eps_i$ are set to 0. We can nevertheless suitably differentiate both numerator and denominator (using L'H\^opital rule) until we reach non-vanishing values, here $\pp_{\eps_N}^{(N-2)}\pp_{\eps_{N-1}}^{(N-3)}..\pp_{\eps_3}\pp_{\eps_2}$. This leads to for $N\ge 4$:
\be
\f{\det\left(e^{i\theta x_jy_k}\right)_{1\le j,k\le N}}{\Delta(Y)}
\quad\underset{\eps_i\arr 0}\longrightarrow\quad
i^{\f{N(N+1)}2}\,\theta^{3(N-3)+1}\,
\f{\sum_\sigma \eps[\sigma]\,x_{\sigma(1)}^{N-2}x_{\sigma(2)}^{N-3}..x_{\sigma(N-2)}\,e^{i\theta x_{\sigma(N-1)}}\,x_{\sigma(N)}e^{i\theta x_{\sigma(N)}}}
{(N-1)!\prod_{k=1}^{N-3}k!}\,.
\ee
The numerator is a modified Vandermonde determinant (but vanishes when $\theta=0$) while the denominator comes from differentiating the original Vandermonde determinant $\Delta(Y)$ (it is also the determinant of the $(N-2)\times(N-2)$ matrix whose matrix elements are given by $m_{ij}=\prod_{k=1}^i (k+j)$).
This provides a direct formula for the observables $\sum_i x_i V_i$ for a diagonal matrix $X$:
$$
\theta\,\sum_i x_i V_i = (\theta\lambda)\,\tr(YU^\dagger X U)
\qquad\textrm{for}\quad
X=(x_1,..x_N)\,.
$$
When the matrix  $X$ is arbitrary and not diagonal, its off-diagonal components allows us to probe the correlations between the various spinors $z_i$:
$$
(\theta\lambda)\,\tr(YU^\dagger X U)
\,=\,
\theta\,\sum_{ij}X_{ij}\la z_i|z_j\ra\,.
$$
Then the Itzykson-Zuber integral can be understood as the generating function for the averages and correlations of the spinor scalar products. From these and taking into account that the vector scalar product is related to the spinor scalar product, $|\la z_i|z_j\ra|^2=V_iV_j+\vV_i\cdot\vV_j$, we can extract in principle all the averages and correlations of the $\SU(2)$-invariant polynomials in the vector components $\vV_i^a$.
It would be interesting to apply these techniques to computing the averages of the powers of the (squared) volume observable, in order to get a better idea of the typical shape of polyhedra, but also because the exact spectrum of the (squared) volume operator at the quantum level is still an open issue.

\medskip

Thus we have seen how the Itzykson-Zuber integral over $\U(N)$ expressed in terms of Vandermonde determinants can be considered as the generating function for the averages of all polynomial observables in the polyhedra's normal vectors. These averages are extracted by suitable differentiating of this Itzykson-Zuber formula. An interesting point is whether the Itzykson-Zuber integrant $e^{i\theta \tr YU^\dagger X U}$ for the fixed considered $Y$ but arbitrary $X$ can have a physical or geometrical relevance, for instance when investigating some (random) dynamics on the space of (framed) polyhedra. We leave this for future investigation.

\subsection{Explicit $\U(N)$ Parametrization and Haar Measure}
\label{explicit}

We now turn to another method to compute these integrals over $\U(N)$ using an explicit parametrization of the unitary matrices and the corresponding recursive formula for the Haar measure on $\U(N)$ \cite{spengler}.

The goal is to draw a unitary matrix at random with respect to the Haar measure, or more precisely to draw at random its two first columns, that is two ortogonal complex $N$-vectors of unit norm. The details of the parametrization and construction for the whole unitary matrix can be found in \cite{spengler}. Here, we will only detail the parametrization of the two first columns and thus of the spinors defining the polyhedra with $N$ faces.

The parametrization is best defined recursively. We start with the case $N=2$. Two arbitrary orthogonal complex $2$-vectors of unit norm can be written as:
\be
v^{(2)}\,=\,
\mat{c}{e^{i\theta_1}\cos\alpha_2 \\ e^{i\theta_2}\sin\alpha_2},
\qquad
w^{(2)}\,=\,
e^{i\phi_2}\,\mat{c}{-e^{i\theta_1}\sin\alpha_2 \\ e^{i\theta_2}\cos\alpha_2}\,,
\ee
where the phases $\theta_1$, $\theta_2$ and $\phi_2$ live in $[0,2\pi]$ while the rotation angle $\alpha_2$'s range is $[0,\f\pi2]$. The normalized Haar measure then reads:
\be
d\mu_2
\,=\,
\f1{\cN_2}\,
\sin(\alpha_2)\cos(\alpha_2)\,d\alpha_2d\theta_1d\theta_2d\phi_2,
\qquad\textrm{with}\quad
\cN_2
\,=\,
\f12\,(2\pi)^3\,.
\ee
The components of the two spinors are read directly from these complex vectors:
$$
z_i=\sqrt{\lambda}\,\mat{c}{v^{(2)}_i \\w^{(2)}_i},
\qquad
z_1=e^{i\theta_1}\,\sqrt{\lambda}\,\mat{c}{\cos\alpha_2 \\-e^{i\phi_2}\,\sin\alpha_2}\,,
z_2=e^{i\theta_2}\,\sqrt{\lambda}\,\mat{c}{\sin\alpha_2 \\ e^{i\phi_2}\,\cos\alpha_2}\,.
$$
This provides a parametrization of a unitary matrix in $\U(2)$ as expected.

Then we can define the two complex vectors $v^{(N)}$ and $w^{(N)}$ recursively from $v^{(N-1)}$ and $w^{(N-1)}$ as:
\be
v^{(N)}
\,=\,
\mat{c}{\cos\alpha_N\,v^{(N-1)} \\ e^{i\theta_N}\sin\alpha_N},
\qquad
w^{(N)}
\,=\,
\mat{c}{\cos\beta_N\,w^{(N-1)} \\ 0}
+e^{i\phi_N}\,\mat{c}{-\sin\alpha_N\sin\beta_N\, v^{(N-1)} \\ e^{i\theta_N}\cos\alpha_N\sin\beta_N}\,,
\ee
where we have added four new parameters, $\theta_N,\phi_N\in[0,2\pi]$ and $\alpha_N,\beta_N\in[0,\f\pi 2]$. The normalized Haar measure is now:
\be
d\mu_N
\,=\,
\f1{\cN_N}\,
d\theta_1\,\,\prod_{k=2}^N \sin\alpha_k \cos^{2k-3}\alpha_k \,d\alpha_kd\theta_kd\phi_k\,\,
\prod_{k=3}^N \sin\beta_k \cos^{2k-5}\beta_k \,d\beta_k\,,
\ee
$$
\qquad\textrm{with}\quad
\cN_n=\f{(2\pi)^{2N-1})}{\prod_{k\ge2} 2(k-1)\,\prod_{k\ge3} 2(k-2)}\,.
$$
%

We can read the components of the $N$ spinors directly from these two complex vectors, up to the global scale factor $\sqrt{\lambda}$. In total, we have parametrized our spinors using $(4N-4)$ angles $\alpha_k,\beta_k,\theta_k,\phi_k$
plus $\lambda$. These are $(4N-3)$ parameters, exactly the dimension of the space of $N$ spinors satisfying the closure constraints.
If we want to further gauge fix the $\SU(2)$ invariance, we can fix the direction of the last vector $\vV_N$. In terms of the components of the last spinor, $z_N=e^{i\theta_N}\,(\sin\alpha_N\,,\,e^{i\phi_N}\,\cos\alpha_N\sin\beta_N)$,
this amounts to fixing $\phi_N=\alpha_N=\beta_N=0$. Fixing these three parameters, this provides an explicit parametrization of the $(4N-6)$-dimensional space $\cP^z_N$ of framed polyhedra up to 3d rotations.

\medskip

If we consider the first vector $v^{(N)}$, we can give its full expression:
\be
v^{(N)}
\,=\,
\mat{c}{e^{i\theta_1}\cos\alpha_2\cos\alpha_3..\cos\alpha_N\\
e^{i\theta_2}\sin\alpha_2\cos\alpha_3..\cos\alpha_N\\
e^{i\theta_3}\sin\alpha_3..\cos\alpha_N\\
\vdots\\
e^{i\theta_N}\cos\alpha_N}
\qquad\textrm{with}\quad
d\mu(v^{(N)})\,\propto\,
\prod_i^Nd\theta_i\prod_{k=2}^N\sin\alpha_k\cos^{2k-3}\alpha_k\,d\alpha_k\,.
\ee
This gives actually a random vector on the complex unit sphere in $\C^N$, distributed uniformly with respect to the Haar measure on $\U(N)$.
It is well known that there is a phenomenon of concentration of measure on the complex sphere as $N$ grows, e.g. \cite{hayden}. More precisely, the integral over the complex sphere is almost equal to the simpler integral over the equator of the sphere (for $\alpha_N=0$). This is due to the specific shape of the Haar measure in this parametrization, which gets concentrated to the equator as $N$ grows large. In the context of quantum information (and quantum computing), this concentration of measure is often used to argue that arbitrary states are maximally entangled between subsystems as the dimensions of the Hilbert spaces grows large, e.g. \cite{hayden,spengler2}.

Here we are drawing a second complex vector  $w^{(N)}$, which is orthogonal to the first one. It would be interesting to investigate whether there is a similar phenomenon of concentration of measure and what would be its geometrical interpretation on the space of (framed) polyhedra. We postpone such analysis to future investigation. Nevertheless, this explicit parametrization does provide a very useful tool in order to compute the average of any polynomial observable over the space of polyhedra as an explicit trigonometric integral.


\section{Deforming Quantum Polyhedra}

This section is dedicated to the study of the quantum case: we quantize the space of framed polyhedra into the Hilbert space of $\SU(2)$ intertwiners interpreted as quantum polyhedra, following the previous work done in \cite{spinor1,un2,un3,un4}. We will see that the Hilbert space of quantum polyhedra has the same structure as the classical set of framed polyhedra. We have indeed a cyclic action of the $\U(N)$ transformations on quantum polyhedra with fixed total boundary area and we can construct coherent polyhedron state labeled by the classical framed polyhedra (up to 3d rotations). Finally, we will give two ways to write the trace of geometrical operators: either using the $\U(N)$ character formula, which is interpreted as the quantum counterpart of the Itzykson-Zuber integral formula or using the coherent states and having an integral over ``fuzzy" polyhedra.

\subsection{Quantizing Polyhedra into Intertwiners}

We canonically quantize the space of spinors $\C^{2N}$ by promoting the components of the spinors and their complex conjugate to harmonic oscillators:
\be
\{z^A_i,\bz^B_j\}
\,=\,
-i\delta^{AB}\delta_{ij}
\quad\longrightarrow\quad
[a^A_i,a^B_j{}^\dagger]
\,=\,
\delta^{AB}\delta_{ij}\,,
\ee
where we have taken the convention $\hbar=1$.
As shown and used in \cite{spinor1,un2,un3,un4} (see also \cite{spinor2,spinor3,hamiltonian3d}), the closure constraints $\cC^a$ generating the $\SU(2)$ action on spinors, the $\U(N)$ generators $E_{ij}$ and the $\SU(2)$-invariant observables $F_{ij}$ are all quantized without ambiguity and their algebra at the quantum level is without any anomaly.
We consistently choose the normal ordering, keeping the annihilation operators $a^{0,1}$ to the right and the creation operators $a^{0,1}{}^\dagger$ to the left.
For details, the interested reader can refer to those references. We will nevertheless give here a quick summary of the main structures, relevant to our main point, that is the $\U(N)$ action on $\SU(2)$ intertwiners.

\medskip

For the sake of completeness, we give the expressions of the basic operators, which are all quadratic in the harmonic oscillators. When there can be no confusion, we will not distinguish the classical quantity from the quantum operator, else we will put a hat $\hat{ }$ on the quantum operator.
For the $\SU(2)$ generators, we have $\cC^a=\sum_iV^a_i$ with:
\be
V^a_i=\sum_{A,B}\sigma_a^{AB}a^{A\dagger}_ia^B_i,
\qquad
V^z_i=(a^{0\dagger}_ia^0_i-a^{1\dagger}_ia^1_i),\quad
V^+_i=a^{0\dagger}_ia^1_i,\quad
V^-_i=a^{1\dagger}_ia^0_i\,.
\ee
These form on each face $i$ the Schwinger representation of the $\su(2)$ algebra in terms of two harmonic oscillators. We also introduce the operator giving the total energy of the oscillators living on the face $i$ as the quantization of the norm of the normal vector $V_i$:
\be
V_i=\sum_{A}a^{A\dagger}_ia^A_i,
\qquad
[V_i\,,\,V_i^a]=0\,.
\ee
As well-known, this $\SU(2)$ representation is reducible and irreducible components are obtained by diagonalizing the Casimir operator $V_i$, whose eigenvalues are twice the spin living on that face, $2j_i\in\N$. This is interpreted as usual as the quantization of the individual face areas.

We similarly quantize the spinor scalar products:
\be
E_{ij}=a^{0\dagger}_ia^0_j+a^{1\dagger}_ia^1_j,
\qquad
F_{ij}=a^0_ia^1_j-a^1_ia^0_j,\qquad
F^\dagger_{ij}=a^{0\dagger}_ia^{1\dagger}_j-a^{1\dagger}_ia^{0\dagger}_j\,.
\ee
It is straightforward to compute the commutators of these operators and check that they give the same results as their Poisson brackets. In particular, the $E$'s and $F$'s commute with the closure constraint operators $\cC^a$ and thus are $\SU(2)$-invariant. Moreover the $E_{ij}$ form a closed $\u(N)$ algebra. Using the definition in terms of harmonic oscillators, the Casimir of this $\u(N)$ algebra is easily related to the total area \cite{un1}:
\be
\sum_{i,j}E_{ij}^\dagger E_{ij}
=E(E+N-1),\qquad
E=\sum_i^NE_{ii}=\sum_i V_i\,.
\ee

\medskip

Looking at the Hilbert spaces, we start with $2N$ copies of the Hilbert space of a single harmonic oscillator, $(\cH_{HO}\otimes \cH_{HO})^{\otimes N}=L^2(\C^{2N})$. Each couple $(\cH_{HO}\otimes \cH_{HO})$ can be decomposed in irreducible representations of $\SU(2)$ with arbitrary spin $j\in\N/2$ (given by half the total number of quanta of the oscillators).
Then we impose a $\SU(2)$-invariance by requiring that the closure constraint operators $\cC^a=\sum_i V^a_i$ vanish on the states. This is exactly the Hilbert space of $\SU(2)$ intertwiners between $N$ irreducible representations:
\be
\cH^{(N)}
\,=\,
\textrm{Inv}_{\SU(2)}\,\Big{[}(\cH_{HO}\otimes \cH_{HO})^{\otimes N}\Big{]}
\,=\,
\textrm{Inv}_{\SU(2)}\,\Big{[}\bigotimes_i^N \bigoplus_{j_i\in\N/2} \cV^{j_i}\Big{]}
\,=\,
\bigoplus_{\{j_i\}_{i=1..N}}\,\textrm{Inv}_{\SU(2)}\,\Big{[}\bigotimes_i^N\cV^{j_i}\Big{]}\,,
\ee
where we write $\cV^j$ for the irreducible $\SU(2)$-representation of spin $j$.
On this Hilbert space of intertwiners, we have a $\U(N)$ action generated by the $E_{ij}$. Since the corresponding $\u(N)$-Casimir $\sum_{i,j}E_{ij}^\dagger E_{ij}$ is determined in terms of the total area operator $E$ whose value is simply the sum of twice the spins $\sum_i^N (2j_i)$, we can simply decompose the space $\cH^{(N)}$ in irreducible components by fixing the value of the total area:
\be
\cH^{(N)}
\,=\,
\bigoplus_{J\in\N} \cR_N^J,\qquad
\cR_N^J
\,=\,
\bigoplus_{\sum^Nj_i=J}
\textrm{Inv}_{\SU(2)}\,\Big{[}\bigotimes_i^N\cV^{j_i}\Big{]}\,,
\ee
where each Hilbert space $\cR^J$ carries an irreducible representation of $\U(N)$, as shown in \cite{spinor1,un1,un2}. The corresponding Young tableaux is given by two horizon lines of equal length $J$. The corresponding highest weight vector $|\psi^J\ra$ corresponds to a bivalent intertwiner, which is the quantum equivalent of the completely squeezed polyhedron in the classical case:
\be
E_{11}\,|\psi^J\ra
\,=\,
J\,|\psi^J\ra,
\quad
E_{22}\,|\psi^J\ra
\,=\,
J\,|\psi^J\ra,
\quad
\forall k\ge3,\,\,E_{kk}\,|\psi^J\ra
\,=\,
0,
\quad
E_{i>j}\,|\psi^J\ra=0\,,
\ee
where the $E_{ii}=V_i$ are the generators of the Cartan subalgebra. In particular, we notice that this highest weight vector is invariant under $\U(N-2)$, which corresponds to the expectation that the classical space of framed polyhedra is isomorphic to the Grassmanniann space $\U(N)/(\U(N-2)\times\SU(2))$.

The dimension of each of these irreducible $\U(N)$-representations can be computed using the hook formula. This gives:
\be
d_N[J]=\dim \cR^J_N =
\,\f1{J+1}\,\mat{c}{N+J-1\\J}\,\mat{c}{N+J-2\\J}\,.
\ee
This is the total number of $\SU(2)$-intertwiners  for a fixed number of faces $N$ and fixed total area $2J=\sum_i 2j_i$. It is the quantum counterpart of the density of states $\rho_N[\lambda]$, which gives the volume of the space of framed polyhedra with $N$ faces and total area $2\lambda$.
Indeed, looking at the large area limit while keeping $N$ fixed, gives:
\be
d_N[J]
\underset{J\arr\infty}{\sim}
\f{J^{2N-4}}{(N-1)!(N-2)!}+\f{NJ^{2N-5}}{(N-1)!(N-3)!}+\dots\,,
\ee
which fits at leading order in $J$ with $\rho_N[\lambda]$, as given by \eqref{rho0}, for $\lambda=J$. Notice that all the terms have the same order in $N$. Therefore, this limit can be considered carefully. To be more rigorous, one should put the $\hbar$-factors back in the quantum expression, then this is the limit where the Planck area unit is sent to 0, while keeping the total area fixed. Then this amounts to sending the sum of the spin to $\infty$, thus giving the wanted result.

\medskip

To summarize the structures, the vector operators $V_i$ acts on each subspace $\cV^{j_i}$ living on each face and generate the $\SU(2)$-action on those subspaces. The $\SU(2)$-invariant operators $E_{ij}$ act on each subspace $\cR_N^J$, defined as the space of $\SU(2)$ intertwiners for fixed sum of the spins $J=\sum_i j_i$, and they generate a $\U(N)$-action on each of these subspaces. Finally the $F_{ij}$ and $F^\dagger_{ij}$ operators respectively act as annihilation and creation operators on the full space of intertwiners $\cH^{(N)}$ allowing to respectively decrease and increase the total area $J$.

These $\SU(2)$-intertwiners are the quantum counterpart of the classical (framed) polyhedra. They are also the basic building blocks of the spin network states of quantum (space) geometry in loop quantum gravity \cite{lqg_review}. That identification of intertwiners as quantum polyhedra is the key to the geometrical  interpretation of spin network as discrete  geometries constructed as (quantum) polyhedra glued together. this identification will be made even clearer below when dealing with coherent intertwiner states peaked on classical framed polyhedra.

\subsection{Beyond Intertwiners: non-Closed Quantum Polyhedra}

Considering the tensor product of $N$ representations of $\SU(2)$, one for to each face of the polyhedron, we have imposed up to now the closure constraint and thus required invariance of our tensor product states under $\SU(2)$. We can relax this condition and characterize states that recouple to a fixed overall spin $\cJ$ different from 0. This corresponds to the classical case where the closure constraints are broken and the sum of the normal vectors do not vanish but the closure vector has a fixed norm.

We are now working on another subspace of $(\cH_{HO}\otimes \cH_{HO})^{\otimes N}=L^2(\C^{2N})$, such that the value of the $\SU(2)$-Casimir given as the norm squared of the closure constraint operators $\cC^2$ is fixed to $\cJ(\cJ+1)$:
\be
\cH^{(N)}_{\cJ}
\,=\,
\textrm{Cov}^{\cJ}_{\SU(2)}\,\Big{[}(\cH_{HO}\otimes \cH_{HO})^{\otimes N}\Big{]}
\,=\,
\bigoplus_{\{j_i\}_{i=1..N}}\,\textrm{Cov}^{\cJ}_{\SU(2)}\,\Big{[}\bigotimes_i^N\cV^{j_i}\Big{]}
\,=\,
\bigoplus_{\{j_i\}_{i=1..N}}\,\textrm{Inv}_{\SU(2)}\,\Big{[}\cV^{\cJ}\otimes\bigotimes_i^N\cV^{j_i}\Big{]}\,.
\ee
This is actually equivalent to having intertwiners, i.e $\SU(2)$-invariant states, between the $N$ original irreducible representations $\cV^{j_i}$ and an extra one $\cV^{\cJ}$.

We still have the $\U(N)$-action on this Hilbert space $\cH^{(N)}_{\cJ}$ and we can decompose it into $\U(N)$ irreducible representations:
\be
\cH^{(N)}_{\cJ}
\,=\,
\sum_{J}\bigoplus_{\sum^Nj_i=J}
\textrm{Cov}^{\cJ}_{\SU(2)}\,
\Big{[}\bigotimes_i^N\cV^{j_i}\Big{]}
\,=\,
\sum_{J}\bigoplus_{\sum^Nj_i=J}
\textrm{Inv}_{\SU(2)}\,\Big{[}\cV^{\cJ}\otimes
\bigotimes_i^N\cV^{j_i}\Big{]}\,,
\ee
where the total area $J$ is of the same parity as the overall spin $\cJ$ (i.e half-integer or integer depending on $\cJ$) and necessarily larger or equal to $\cJ$.

Each of the subspaces at fixed $J$ carries an irreducible representation of $\U(N)$. Its highest weight vector is defined by the (unique) trivalent intertwiner between $\SU(2)$-representations of spins $\f{J+\cJ}2$, $\f{J-\cJ}2$ and $\cJ$, i.e the values of the Cartan subalgebra generators on it are \cite{un1}:
$$
E_{11}\,|\psi^J_\cJ\ra
\,=\,
(J+\cJ)\,|\psi^J_\cJ\ra,
\quad
E_{22}\,|\psi^J_\cJ\ra
\,=\,
(J-\cJ)\,|\psi^J_\cJ\ra,
\quad
\forall k\ge3,\,\,E_{kk}\,|\psi^J_\cJ\ra
\,=\,
0\,.
$$
Thus the corresponding Young tableaux contains two horizontal lines of respective lengths $(J+\cJ)$ and $(J-\cJ)$ and the dimensions of the representations are \cite{un1}:
\be
d_N[J,\cJ]
\,=\,
\dim  \cR^{J,\cJ}_N
\,=\,
\dim \sum_{J}\bigoplus_{\sum^Nj_i=J}
\textrm{Cov}^{\cJ}_{\SU(2)}\,
\Big{[}\bigotimes_i^N\cV^{j_i}\Big{]}
\,=\,
\f{2\cJ+1}{J+\cJ+1}\,\mat{c}{N+J+\cJ-1\\J+\cJ}\,\mat{c}{N+J-\cJ-2\\J-\cJ}\,.
\ee
It is fairly easy to check that summing over all possible values of $\cJ\le J$, we recover the full Hilbert space of intertwiners for $(N+1)$ faces and fixed total area $J$:
\be
d_{N+1}[J]
\,=\,
\sum_{\cJ\le J}
d_N[J,\cJ]\,.
\ee
This could be proved directly either by recombining the binomial coefficients or by recursion.

\medskip

Finally, it would be interesting to investigate whether there is a similar procedure to ``close" non-invariant configuration as in the classical case, where we could apply a $\SL(2,\C)$ transformations on an arbitrary non-closed set of spinors in order to map it into a closed set of spinors defining an actual framed polyhedron. We postpone to future investigation the thorough study of the existence on a $\SL(2,\C)$-action on the space of intertwiners and of its properties.

\subsection{Probing the shape of Intertwiners}

Similarly to the classical case, we now would like to compute the traces of geometrical operators on the Hilbert space of $\SU(2)$-intertwiners at fixed number $N$ of faces and fixed total area $J=\sum_i j_i$.

We can already deduce some averages from the fixed area condition $J=\sum_i j_i$ and the formula for the dimensions of the intertwiner spaces $d_N[J]$. We obviously have:
\be
\la 2j_i\ra=\la V_i \ra
\,\equiv\,
\f1{d_N[J]}\,\tr_{\cH^{(N)}} V_i
=\f{2J}N\,,
\ee
which is also equal to the classical average \eqref{Vcl}.
We can also single out explicitly one face/leg of the intertwiner. Then using the dimension of the space of tensor product states for a fixed external spin (or intertwiners with one fixed spin) given in the previous section, we compute:
\be
\la V_i^2 \ra = \la 4j_i(j_i+1) \ra
=\f{1}{d_N[J]}\sum_{j\ge \f J2} 4j(j+1)d_{N-1}[J,j]
=\f{6J(J+N)}{N(N+1)}\
=\f{6J^2}{N(N+1)}+\f{6J}{N+1}
,.
\ee
We see that the first term in $\la V_i^2 \ra$ fits exactly the classical average \eqref{V2cl}. The second term is the quantum correction, and is sub-leading in the classical limit defined by taking large $J$ at fixed $N$.


Playing around with the binomial coefficients, one can show the somewhat surprising formula giving the traces of  arbitrary powers of the norm:
\beq
\la 2j_i(2j_i+1)..(2j_i+n) \ra
&=&
\f{1}{d_N[J]}\sum_{j\ge \f J2} 2j(2j+1)..(2j+n)d_{N-1}[J,j]\nn\\
&=&
J((m+2)J+2N+m-2)\,(m+1)!\,\f{(N+J+m-2)!}{(N+J-1)!}\f{(N-1)!}{(N+m)!} \\
&=&
J((m+2)J+2N+m-2)\,(m+1)!\,
\f{\mat{c}{N+J+m-2\\m-1}}
{\mat{c}{N+m\\m-1}}\,,\nn
\eeq
from which we can recover the traces $\la V_i\ra$ and $\la V_i^2\ra$.

\medskip

We can square the fixed area condition and deduce the correlation between spins $i\ne k$:
\beq
\left\la\left(\sum_i V_i\right)^2\right\ra
&=&N\,\la V_i^2\ra+ N(N-1)\,\la V_iV_k\ra_{i\ne k}=(2J)^2 \nn\\
&&\qquad\Rightarrow\quad
\la V_iV_k\ra_{i\ne k}=\la 4j_ij_k\ra
=J^2\,\f{2(2N-1)}{(N-1)N(N+1)}-\f{6J}{(N-1)(N+1)}\,.
\eeq
Similarly, using the closure constraint operator, or in other words the $\SU(2)$-invariance, we can compute:
\be
\left\la\left(\sum_i \vV_i\right)^2\right\ra
=N\,\la V_i^2\ra+ N(N-1)\,\la \vV_i\cdot\vV_k\ra_{i\ne k}=0
\qquad\Rightarrow\quad
\la \vV_i\cdot\vV_k\ra_{i\ne k}=
\f{-6J(J+N)}{(N-1)N(N+1)}\,.
\ee
If we want to go further and compute traces of operators involving the values of the spins on three or more legs and thus probing the fine structure of the intertwiners, we would have to compute the dimensions of the intertwiner subspaces with fixed spins. Instead of doing this by hand, we can do this consistently using the full $\U(N)$-character formula, which computes the trace of $\U(N)$ transformations instead of simply the dimension which gives the trace of the identity. This the method outlined in \cite{un1} and we show here that it should be considered as a generalization to the quantum case of the Itzykson-Zuber formula used as generating function for averages over the ensemble of classical polyhedra.

\medskip

More precisely, the character of the $\U(N)$ representation, of highest weight $[l_1,..l_N]$, computes the trace of a diagonalized unitary transformation $U=(e^{i\theta_1},..,e^{i\theta_N})$ as a Schur polynomial:
\be
\chi_{[l_i]}(e^{i\theta_i})
\,=\,
\f{\det (e^{i\theta_j\,(l_i+N-i)})_{ij} }{\det (e^{i\theta_j\,(N-i)})_{ij}}\,.
\ee
Here, the highest weight is given by $l_1=l_2=J$ and this formula defines directly the generating functions for the spin expectation values (or equivalently the $V_i=2j_i$):
\be
\la e^{i\sum_k\theta_k E_k}\ra
\,=\,
\f{\chi_{[J,J,0,..]}(e^{i\theta_k})}{d_N[J]}
\quad=\,
\f1{d_N[J]}\,
\f{\det (e^{i\theta_j\,(J(\delta_{i1}+\delta_{i2})+N-i)})_{ij} }{\det (e^{i\theta_j\,(N-i)})_{ij}}\,,
\ee
where the normalization should be such that the expectation value of 1 is 1 (when $\theta_i=0$). The determinant at the denominator is exactly a Vandermonde determinant, while the numerator is a slight modification.

This formula contains all the traces of polynomials in the spins $j_i$'s. If we extend the formula to non-diagonal $\U(N)$ transformations (which we can diagonalize of course), we can generate the traces of all scalar products and powers in the basic vectors $\vV_i$.
As in the classical case, extracting these traces requires a careful differentiation of this generating function. It would be interesting if these traces of $\U(N)$ transformations could themselves be physically/geometrically relevant, for instance in the study of the dynamics of polyhedra and intertwiners.

%

\subsection{Interpolating between Classical and Quantum Polyhedra: Coherent Intertwiner States}

To better understand the link between intertwined states and classical polyhedra, we can build coherent intertwined states peaked on classical framed polyhedra following \cite{un2, un3,un4}. 
Following the conventions of \cite{un3,un4}, one defines:

\begin{definition}
Given a set of spinors $z_{i}\in\C^{2N}$, we define the coherent intertwiner state $|J,\{z_{i}\}\ra$ in $\cR^J_{N}$ using the $\SU(2)$ creation operators $F^{\dagger}$:
\be
|J,\{z_{i}\}\ra
\,=\,
\f1{\sqrt{J!(J+1)!}}\,\left(
\f12 \sum_{ij}[z_{i}|z_{j}\ra\,F^{\dagger}_{ij}
\right)^J
\,|0\ra
\,=\,
\f1{\sqrt{J!(J+1)!}}\,\left(
\f12 \sum_{ij}[z_{i}|z_{j}\ra\,(a^{0\dagger}_{i}a^{1\dagger}_{j}-a^{1\dagger}_{i}a^{0\dagger}_{j})
\right)^J
\,|0\ra\,.
\ee
The scalar products $[z_{i}|z_{j}\ra$ are invariant under $\SU(2)$ rotations, so the intertwiner states are labeled by the orbits of spinors under global $\SU(2)$ transformations. Moreover, these scalar products are also invariant under global $\SL(2,\C)$ transformations, which map arbitrary sets of spinors to spinors satisfying the closure constraints. Thus the coherent intertwiner states are truly labeled by orbits of spinors under global $\SL(2,\C)$ transformations, that is points in the space of framed polyhedra (up to 3d rotations) $\cP^z_{N}=\C^{2N}/\SL(2,\C)=\C^{2N}//\SU(2)$ as we have seen in section \ref{param}.

\end{definition}

The main results established in \cite{un2}, and revisited in \cite{un3,un4}, are two key properties of these intertwiner coherent states: their formulation as group averaging of the tensor product of standard $\SU(2)$ coherent states, which establishes their geometrical interpretation as semi-classical polyhedron states, and then their coherence under the action of $\U(N)$. Or more precisely,

\begin{itemize}

\item {\bf Decomposition on $\SU(2)$ coherent states:}

\be
\f1{\sqrt{J!(J+1)!}}\,|J,\{z_{i}\}\ra
\,=\,
\sum_{J=\sum_{i}j_{i}}
\f1{\prod_{i}(2j_{i})!}
\int_{\SU(2)}dg\, \bigotimes_{i}g\,|j_{i},z_{i}\ra\,,
\ee
where we group average the tensor product of individual $\SU(2)$ coherent states living on each face and defined as:
$$
|j,z\ra
\,=\,
\f{(z^0a^{0\dagger}+z^1a^{1\dagger})^{2j}}{\sqrt{(2j)!}}\,|0\ra.
$$
These states living in $\cV^j$ are coherent under the action of $\SU(2)$ and thus can all be generated from the highest weight vector $|j,j\ra$ by acting with $\SU(2)$ transformations (up to a norm factor):
\be
g\,|j,z\ra\,=\,|j,gz\ra,
\qquad
\f1{\sqrt{\la z|z\ra^{2j}}}\,|j,z\ra
\,=\,
g(z)\,|j,j\ra,
\quad\textrm{with}\quad
g(z)
\,=\,
\f1{\sqrt{\la z|z\ra}}\,\mat{cc}{z^0 & -\bz^1 \\ z^1 & \bz^0}\,.
\ee
Finally, they are peaked on the classical vectors $\vV(z)=\la z|\vsigma|z\ra$:
\be
\f{\la j,z| \hat{V}^a |j,z\ra}{\la j,z|j,z\ra}
\,=\,
2j\,\f{\vV}{V}\,,
\ee
where the expectation value vector has the same direction as $\vV$ but is normalized to $2j$ in term of the spin carried by the state. The group average states for fixed individual spins $j_{i}$ were introduced earlier in \cite{intertwiner}.

Written as such, the coherent intertwiner states $|J,\{z_{i}\}\ra$ truly represent the quantized version of a classical framed polyhedron defined as a set of $N$ vectors or spinors up to $\SU(2)$ transformations.

\item {\bf Coherence under $\U(N)$ transformations:}

The action of $\U(N)$ transformations, generated by the operators $\hat{E}_{ij}$ at the quantum level, on the coherent intertwiner states amounts to the classical $\U(N)$-action on the set of spinors labeling the state:
\be
\hat{U}\,|J,{z_{i}}\ra
\,=\,
|J,{(Uz)_{i}}\ra,
\quad
U=e^{i h},
\quad
\hat{U}=e^{i\sum_{kl}h_{kl}\hat{E}_{kl}}\,,
\ee
for an arbitrary Hermitian matrix $h$. This ensures that the behavior of coherent intertwiner states is just the same as classical framed polyhedra. For instance, one can generate all coherent intertwines by acting with $\U(N)$ transformations on the bivalent intertwiner, just the same way as we could generate all (closed) framed polyhedra by acting with $\U(N)$ transformations on the totally squeezed configuration with only two non-trivial faces. This is the key property allowing us to take the trace over the Hilbert space of intertwines by an integral over the unitarity group $\U(N)$, similarly to the classical case. This is explained below in details.

\end{itemize}

Taking into account that the Hilbert space $\cR^J_{N}$ of intertwiners for fixed total sum of the spins is an irreducible representation of $\U(N)$, one can write the identity of that space as an integral over $\U(N)$ acting on a fixed state, say the bivalent intertwiner on the legs 1 and 2, which is exactly the integral over the coherent intertwiner states: 
\be
\id^J_N
\,=\,
\f1{J!(J+1)!}\,
\int_{\C^{2N}}
\prod_i\f{e^{-\la z_i|z_i\ra}d^4z_i}{\pi^2}\,
|J,\{z_i\}\ra\la J,\{z_i\}|\,.
\ee
A rigorous proof can be found in \cite{un2}, and then a simpler proof in \cite{un3,un4}. Basically, this comes from writing the identity on the larger Hilbert space $\cH^{(N)}$ in terms of the usual coherent states for the harmonic oscillators  and then projecting down on the subspace with fixed total number of quanta $J$.

We also compute the scalar product between two coherent intertwiners \cite{un2}:
\be
\la J,\{z_i\}|J,\{w_i\}\ra
\,=\,
\left(\det\sum_i|w_i\ra\la z_i|\right)^J
\,=\,
\left(\f12\sum_{i,j}[w_i|w_j\ra\la z_j|z_i]\right)^J\,.
\ee
For a single set of spinors, this also gives the norm of the coherent intertwiner state:
\be
\la J,\{z_i\}|J,\{z_i\}\ra
\,=\,
\left(\f12\sum_{i,j}|F_{ij}|^2\right)^J
\,=\,
\f1{2^{2J}}\,\left(\left(\sum_i V_i\right)^2-\left|\sum_i \vV_i\right|^2\right)^J\,.
\ee
We clearly see that the norm is maximized when the closure vector vanishes, $\vcC=\sum_i \vV_i=0$, and that we have a closed set of spinors thus corresponding to a true polyhedron. This is similarly to the result obtained in \cite{intertwiner}, which studied the saddle point approximation to the group averaging for fixed spins $j_{i}$ and show that a stationary point exists only if the vectors satisfy the closure constraints.

Combining these two formula, we can take the trace of the identity on $\cR^J_{N}$ and recover the dimension of this Hilbert space, which we can express either as a Gaussian integral over the spinors $z_{i}$ or as an integral over the vectors $\vV_{i}$:
\beq
d_N[J]=\tr \id^J_N
&=&
\f1{J!(J+1)!}\,
\int_{\C^{2N}}
\prod_i\f{e^{-\la z_i|z_i\ra}d^4z_i}{\pi^2}\,
\left(\det\sum_i|z_i\ra\la z_i|\right)^J \nn\\
&=&
\f1{2^{2J} J!(J+1)!}\,
\int_{\R^{3N}}\prod_i\f{e^{-V_i}d^3\vV_i}{4\pi V_i}\,
\left(\left(\sum_i V_i\right)^2-\left|\sum_i \vV_i\right|^2\right)^J\,.
\eeq
It is possible to check this formula directly by performing the integral over the vectors $\vV_{i}$ as done in \cite{counting}. This provides a interpretation of the space of intertwiners as almost-closed polyhedra, or fuzzy polyhedra, which are more and more peaked on true framed polyhedra (satisfying the closure constraints) as the total area $J$ grows. This fits with our earlier claim that the dimension of the Hilbert space $\cR^J_{N}$ behaves at leading order as the classical density of framed polyhedra as the total spin $J$ grows large for a fixed number of faces $N$.

We can know further the compute the trace of any unitary transformation, which provides an integral formula for the $\U(N)$-character given in terms of modified Vandermonde determinant in the previous section:
\beq
\chi_{[J,J,0,..]}(\hat{U})
=\tr^J_{N} \hat{U}
&=&
\f1{J!(J+1)!}\,
\int_{\C^{2N}}
\prod_i\f{e^{-\la z_i|z_i\ra}d^4z_i}{\pi^2}\,
\la J,{z_{i}}|J,{(Uz)_{i}}\ra \nn\\
&=&
\f1{J!(J+1)!}\,
\int_{\C^{2N}}
\prod_i\f{e^{-\la z_i|z_i\ra}d^4z_i}{\pi^2}\,
\left(\det\sum_i|(Uz)_i\ra\la z_i|\right)^J\,,
\eeq
with $U=e^{i h}$ and $\hat{U}=e^{i\sum_{kl}h_{kl}\hat{E}_{kl}}$ in terms of a Hermitian matrix $h$ as above.
In general, the determinant is a little messy:
$$
\det\sum_i|(Uz)_i\ra\la z_i|
=\f12 \sum_{ijkl}U_{ik}U_{jl}F_{kl}\bF_{ij}\,,
$$
from which we can compute the trace of $\hat{U}$ as Gaussian integral in the spinor variables.
It is however simpler when looking at diagonal unitary transformations $U=e^{i\sum_k \theta_kE_k}$, which act as multiplication by individual phases on each spinor:
\be
\la J,\{z_i\}| e^{i\sum_k \theta_kE_k}|J,\{z_i\}\ra
=
\la J,\{z_i\}|J,\{e^{i\theta_i}z_i\}\ra
=
\left[\f12\sum_{ij}e^{i(\theta_{i}+\theta_{j})}|F_{ij}|^2\right]^{J}
=
\f1{2^{2J}}\,\left(\left(\sum_i e^{i\theta_{i}}V_i\right)^2-\left|\sum_i e^{i\theta_{i}}\,\vV_i\right|^2\right)^{J}\,,\nn
\ee
from which we can compute the trace of $e^{i\sum_k \theta_kE_k}$ as an integral over the spinors or the vectors.

These traces are to be considered as the generating function for the traces of every polynomial operators in the $E$'s or $F$'s or $V$'s. Beyond this, it would be interesting to investigate if these unitary transformations can be seen as implementing the dynamics of the intertwiners for fixed boundary area (for instance, in the context of quantum black holes, see e.g. \cite{counting}), in which case these traces could be provided with a direct physical interpretation. On the other hand, we could also apply our machinery to other geometric operators. For example, an open issue is still to compute the exact spectrum of the (squared) volume operator (see nevertheless e.g. \cite{volumelqg}) and we could attempt to compute the traces of the powers of this operator, from which we could reconstruct its spectrum.

\section{Orthogonal Group Action on Polygons}

We will now look at structures in one dimension less and study the space of polygons. We will see that we can similarly define a phase space of framed polygons, also the frame now on each edge will be reduced to a sign $\pm$. Similarly to the case of polyhedra, we will characterize the space of framed polygons for fixed boundary perimeter as a representation of the orthogonal group $\O(N)$ instead of the unitary group. Working in 2d instead of 3d will also us to be more explicit in the reconstruction of the geometrical structure especially on the issue of gluing polygons together to form a two-dimensional discrete manifold. 

\subsection{Phase Space for Polygons}

The phase space structure for polygons is much simpler than for polyhedra. Instead of using spinors attached to each face of the polyhedron, we will attach a single complex variable to each edge of the polygon. Let us thus start with $\{z_i\}\in\C^{N}$ with $i=1..N$ for polygons with $N$ edges and postulate the following canonical Poisson bracket:
\be
\{z_j,\bz_k\}=-i\delta_{jk}.
\ee
This corresponds to a set of $N$ oscillators. Decomposing the complex variables in real and imaginary parts, $z_{j}=R_{j}+iI_{j}$, these variables look like the real version of the spinors used in the 3d case for polyhedra:
$$
z_{j}=R_{j}+iI_{j}\,\in\,\C
\quad\longrightarrow\quad
\mat{c}{R_{j}\\ I_{j}}\,\in\,\R^{2}\,.
$$
In these real variables, the canonical bracket reads $\{R_j,I_k\}=\f12\,\delta_{jk}$.
We define a closure constraint to ensure that the complex variables correspond the normal vectors to the edges of a true closed polygon:
\be
\cC=\sum_j z_j^2\,.
\ee
The normal vectors are 2-dimensional and correspond to the square of the complex variables:
\be
z_j^2=(R_j^2-I_j^2)+2iR_jI_j
\quad\longrightarrow\quad
\vn_{j}=
\mat{c}{R_j^2-I_j^2\\ 2R_jI_j}\,\in\,\R^{2}\,.
\ee
As we will see in details in a following section \ref{polygongeom}, this ensures that we can reconstruct a unique convex polygon, such that these normal vectors are orthogonal to the polygon's edges and their norm give the length of the corresponding edge.
Then  the perimeter of the polygon is given by the total energy of the oscillators:
\be
\cE=\sum_j |z_j|^2\,.
\ee

The real part of the closure constraint generates the multiplication by a global $\U(1)$ phase to all the complex variables:
\be
e^{i\theta\,\{\sum_{k}(R_k^2-I_k^2),\cdot\}}\,\mat{c}{R_{j}\\ I_{j}}
\,=\,
\mat{cc}{\cos\theta &\sin\theta\\-\sin\theta & \cos\theta}\,\mat{c}{R_{j}\\ I_{j}}\,,
\qquad
e^{i\theta\,\{\sum_{k}(R_k^2-I_k^2),\cdot\}}\,z_{j}
=e^{i\theta}\,z_{j}\,,
\ee
with $\theta\in\R$ and $e^{i\theta}\in\U(1)$.
The closure constraint $\cC=0$ is clearly invariant under this global phase transformation.

On the other hand, the imaginary part of $\cC$ generates global inverse re-scaling of the real and imaginary parts:
\be
e^{-\eta\,\{\sum_{k}2R_kI_k,\cdot\}}\,\mat{c}{R_{j}\\ I_{j}}
\,=\,
\mat{cc}{e^{\eta} &0\\0 & e^{-\eta}}\,\mat{c}{R_{j}\\ I_{j}}\,,
\ee
with $\eta\in\R$.
The closure constraint $\cC=0$ is not invariant under such transformations. On the contrary, we can use them map any set of complex variables onto a closed set satisfying $\cC=0$. Indeed, starting with an arbitrary value of $\sum_{j}z_{j}^{2}$, we first set the phase of this complex number by multiplication by a global phase. It is then purely real. Second, we set its real part to 0 by the inverse re-scaling which allows to balance the sum of the squares of the real parts and imaginary parts, $\sum_{k}R_k^2 =\sum_{k}I_k^2$.

Combining these two type of transformations generate the $\SL(2,\R)$ group. As we have just shown, these transformations allow to map any complex $N$-vector $(z_{k})_{k=1..N}\in\C^{N}$ onto one satisfying the closure constraint $\cC=0$. This is the equivalent of the $\SL(2,\C)$ transformations allowing to map arbitrary sets of $N$ spinors onto a closed framed polyhedron.

\subsection{The Orthogonal Group Action}

As before, we have the obvious $\U(N)$-action on $\C^N$ now generated by $E_{jk}=\bz_j z_k$:
\be
z_i
\,\longrightarrow\,
(Uz)_i=\sum_j U_{ij}z_j\,.
\ee
It allows to go from $\Om=(1,0,..,0)$ to an arbitrary vector in $\C^N$ up to a global scale factor:
\be
\Om\,\longrightarrow\, (U\Om)_i=U_{i1}\,,
\ee
which means that we are working on the unit complex sphere $U(N)/\U(N-1)\sim\cS^\C_{N-1}\sim\cS^\R_{2N-1}$. This could be an interesting testing ground for the case of the polyhedra since we know well the phenomenon of concentration of measure on the coset $U(N)/\U(N-1)$ as $N$ grows to infinity.
As expected, the $\U(N)$ action leaves invariant the perimeter $\cE$:
\be
\{E_{ij},\cE\}=0\,.
\ee
On the other hand, it does not commute with the closure constraint. However, we can introduce a linear combination of the $\u(N)$ generators that does:
\be
\{E_{ij},\cC\}=iz_jz_i\ne0,\qquad
\{e_{ij},\cC\}=0
\quad\textrm{with}\quad
e_{ij}\equiv
-i(E_{ij}-E_{ji})=-i(\bz_i z_j -z_i\bz_j)\,.
\ee
These form a $\o(N)$ algebra and actually generate the following action of $\O(N)$ on the complex $N$-vector:
\be
z_i
\,\longrightarrow\,
(Oz)_i=\sum_j O_{ij}z_j\,.
\ee
It leaves invariant the perimeter and the closure constraint:
\beq
\sum_i |z_i|^2
&\longrightarrow&
\sum_{jk} \sum_i O_{ij}\bz_j \,O_{ik}z_k
\,=\,
\sum_{jk} \delta_{jk}\bz_j z_k=\sum_i |z_i|^2\,, \nn\\
\sum_i z_i^2
&\longrightarrow&
\sum_{jk} \sum_i O_{ij}z_j \,O_{ik}z_k
\,=\,
\sum_{jk} \delta_{jk}z_j z_k=\sum_i z_i^2\,. \nn
\eeq
It is interesting that this action is cyclic on the set of vectors satisfying the closure constraint. Indeed starting with the vector $\om=(1,i,0,..0)$ with ``unit" perimeter, $\cE=2$, and trivially satisfying the closure constraint, we perform an orthogonal transformation $O$, with $O_{ij}\in\R$ and ${}^tOO=\id$:
\be
\om_i\,\longrightarrow\,
(O\om)_i=O_{i1}+i O_{i2}\,.
\ee
Thus the orthogonal matrix gives the real and imaginary parts of the complex variables. Reciprocally, starting with a $N$-vector with unit perimeter $\cE=2$, we write both the fixed perimeter and closure constraints in terms of the real and imaginary parts of the complex coordinates, $z_i=R_i+iI_i$:
\be
\sum_i z_i^2=0
\,=\,
\sum_i (R_i^2-I_i^2)+2i\sum_i R_iI_i\,,
\ee
\be
\sum_i |z_i|^2=0
\,=\,
\sum_i (R_i^2+I_i^2)\,,
\ee
which mean that the real $N$-dimensional vectors $R_i$ and $I_i$ are orthonormal, and thus can be identified as the first two columns of an orthogonal matrix, $R_i=O_{i1}$ and $I_i=O_{i2}$.

At the end of the day, we will be able to describe averages on the ensemble of polygons as integrals over the orthogonal group. We will go further in this direction, although we can compute similarly to the unitary group case polynomial integrals and a Itzykson-Zuber formula over $\O(N)$. Instead we will focus on the geometrical interpretation of this phase space.

\subsection{Reconstructing Polygons}
\label{polygongeom}

Let us describe how we actually go from our complex $N$-dimensional vector $z_i$ satisfying the closure constraint to a real closed polygon (embedded in the flat plane).
We would like to interpret the complex variable $z_i$ as defining the normal to an edge of the polygon. More precisely, we identify the 2-vector $\vn_i\in\R^2$ normal to the edge $i$ to $z_i^2\in\C=\R^2$ with the edge length given by the modulus square  $l_i=|z_i|^2$.

The crucial step of the reconstruction is that we need to (re)-order the edges according to the angle of the normal vector $\vn_i$, or equivalently to the phase of $z_i^2$, so that the angles taken between 0 and $2\pi$ grows with the edge label $i$.
Starting arbitrarily the position $\vv_{1}$ of the first vertex of the polygon, say on the positive real axis for the sake of simplicity, we reconstruct the next vertex positions $\vv_{i}$  from $\vn_i=(\vv_{i+1}-\vv_i)\w \he_z$, or equivalently $(\vv_{i+1}-\vv_i)=\vn_i\w\he_z$, where $\he_{z}$ is the axis orthogonal to the plane. The closure constraint $\sum \vn_i=0$, equivalent to  $\sum z_i^2=0$, ensures that this procedure defines an actual polygon, with $\vv_{N+1}=\vv_{1}$.
Then we would like to first check that our polygon is convex, i.e. that the angle between two consecutive displacement vectors, or equivalently two consecutive normal vectors, is always at most 180 degrees. Mathematically, this translates to $\he_z\cdot[(\vv_{i+1}-\vv_i)\w(\vv_{i+1}-\vv_{i-1})]\ge 0$ for all $i$'s, or equivalently $\he_z\cdot(\vn_i\w\vn_{i-1})\ge 0$. Since we have ordered all the normals with growing angles, this convexity condition is automatically fulfilled, else the closure constraint can not be satisfied. This concludes he reconstruction procedure for the polygon, which is significantly simpler than for the polyhedron (see e.g. \cite{dona}).

\medskip

An interesting feature of our phase space construction is that the normal vectors, and thus the actual geometric polygon, is invariant under the change of sign of individual complex variables $z_i\arr -z_i$. This sign is nevertheless relevant when looking at the action of the orthogonal group on the polygons, i.e two sets of complex variables differing by signs but defining the same polygon will have different images under an orthogonal transformation. This sign ``ambiguity'' is the equivalent of the phase of the spinor variables for the polyhedra. Then we similarly introduce the notion of ``signed" polygons, corresponding to ``framed" polyhedra in the 3d case.
We expect this sign to be relevant when gluing the polygons together, just as the spinor phases played an essential role when gluing (framed) polyhedra into twisted geometries (encoding the Ashtekar-Barbero connection along the edge \cite{twisted1}). Let us look a bit more into this in the next section.


\subsection{Deforming and Gluing Polygons}

Similarly to the spinor networks introduced as the classical phase space underlying the spin network of loop quantum gravity on a fixed graph \cite{spinor1,spinor2,twisted1,twisted2} and interpreted as twisted geometries, we would like to introduce its two-dimensional equivalent, corresponding to gluing polygons along a given graph. Let us consider an abstract (oriented) closed graph $\Gamma$. Around each vertex $v$ of the graph, we will consider one complex variable $z^{v}_{l}$ for each link $l$ attached to $v$. Reciprocally, for each link $l$ of the graph, we will have two complex variables $z^{s,t}_{l}$ for the two vertices bounding $l$, for its source and target vertices $v=s,t(l)$. We assume the canonical Poisson bracket for each complex variable, plus one closure constraint at each vertex, plus one length matching constraint on each link:
\be
\{z^{v}_{l},\bz^{v}_{l}\}=-i,
\qquad
\forall v,\,\,\,\sum_{l\ni v} (z^{v}_{l})^2=0,
\qquad
\forall l,\,\,\,|z^{s}_{l}|^{2}=|z^{t}_{l}|^{2}\,.
\ee
Geometrically, we have one polygon dual to its vertex. These polygons are then glued together edge by edge along each link of the graph. We call this a ``complex network'', where complex stands for the complex variables used on each edge instead of spinors.
Let us emphasize that although each polygon is constructed in a fixed plane as a purely two-dimensional object, the glued polygons are not to be thought as in the same plane. Indeed, we can think of each polygon as in its own tangent plane to the overall 2d discrete manifold, with its normal vectors defined in that tangent plane. 

This can be considered as a toy model for the gluing of polyhedra in 3d and the study of the deformation and dynamics of twisted geometries. There is no shape matching problem as in 3d, where we have an area matching between polyhedra ensuring that the two faces to be glued have the same area but not necessarily the same shape. But we have nevertheless the issue of reconstructing globally the dual cellular complex (i.e the ``triangulation''). A first look easily shows the problems. Let us start with a cellular complex for a 2d manifold, as a set of flat convex polygons glued together, and we consider the graph defined as its dual 1-skeleton and the corresponding complex network living on it encoding the geometrical data of the polygons. If we start modifying the normal vectors around the vertices, still making sure of not changing the closure constraints and the length matching constraints, we can change the angles of the normal vectors around each vertex and nothing a priori ensures that the ordering of the edges remains consistent to the original one and defines the same cellular complex as before. This seems to imply that deforming the complex network can induce a global change of the dual cellular complex (definition of the points dual to the faces/loops of the network) and probably of its topology. An alternative would be to fix the ordering of the edges around each vertex and not modify it while deforming the angles and norms of the normal vectors, thus allowing for the reconstruction of non-convex polygons. We face the same issue(s) in 3d considering the deformations of glued polyhedra and it would probably enlightening to explore the various possibilities and solve these problems in the 2d case studying the dynamics of glued polygons.

From this perspective, we plan to report in a separate paper the analysis of the dynamics of these complex networks and the issue of deforming the gluing of polygons. This should involve introducing some Hamiltonian constraints imposing some flatness conditions on the glued polygons and studying the dynamics of the 2d geometry induced by these constraints. Particular care should taken in understanding the role (if any) of the sign ambiguity between the complex variables $z$'s and the normal vectors $\vn$'s.


%
%
%

\section{Outlook: Matrix Models for Dynamical Polyhedra}
\label{outlook}

We would like to finish this paper on the possibility of defining and studying the dynamics of framed polyhedra in the $\U(N)$ framework presented here.
We first would like to define the kinetic term, encoding the dynamical degrees of freedom and their Poisson bracket. As the spinor variables have canonical brackets, it is natural to postulate the straightforward kinetic term for them, as assumed in \cite{spinor1}. Then keeping in mind the definition of the spinors in terms of the $\U(N)$ matrix $U$ and the total boundary area $2\lambda$, $z_{i}^{A}=\lambda\,U_{iA}$ for the face index $i=1..N$ and the spinor index $A=1,2$, we can express the kinetic term entirely in terms of the unitary matrix:
\be
S_{kin}
=\int dt\, -i\sum_{k}\la z_{k}|\pp_{t}z_{k} \ra
=\int dt\, -i\lambda\tr YU^{\dagger}\pp_{t}U
=\int dt\, +i\lambda\tr UY\pp_{t}U^{\dagger},
\qquad\textrm{with}\quad
Y=\mat{c|c}{\id_2&\\
\hline&0_{N-2}}\,.
\ee
We then have to define a Hamiltonian and potential. We can not require $\U(N)$ invariance as we would naturally do when dealing with matrix models else our model would collapse to a pure isotropic behavior independent of the unitary matrix $U$ and described only by the dynamics of the total boundary area $\lambda$. If we want some dynamics deforming the shape of the polyhedron, a natural possibility\footnotemark to explore is to introduce an external source given for example as a Hermitian matrix $X$ with a non-trivial potential and define the full action in terms of the unitary matrix $U$:
\be
S[U]
\,=\,
\int dt\, \left(-i\lambda\tr YU^{\dagger}\pp_{t}U - \lambda \tr YU^\dagger XU +V[X]\right)\,.
\ee
The equations of motion are straightforward to compute:
\be
\lambda UYU^{\dagger}=\pp_{X}V,
\qquad
(i\pp_{t}U +XU)Y=0\,,
\ee
with the equation of motion for $\lambda$ being trivial. The potential $V[X]$ should not be taken $\U(N)$-invariant, else the theory would be invariant under the action of the unitary group and thus trivial (with only the global area being dynamical). We should investigate how to choose a physically-relevant potential, for example in relation to cosmological mini-superspaces in quantum gravity (e.g. \cite{sfcosmo}). 

This would model the evolution of a given polyhedron, within a twisted geometry, coupled to some external excitation a priori taking into account the interaction of the polyhedron with the rest of the geometry. At the quantum level, this would model the dynamics of an intertwiner with the outside geometry in the context of loop quantum gravity. It would thus be interesting to solve these equations and see the various behavior of the evolution of $U$ in terms of the choice of potential $V[X]$.
\footnotetext{An alternative to a fixed potential would be to have a dynamical source $X$ with its own kinetic term and a non-trivial interaction. Such a model has been defined in \cite{spinor1} with two polyhedra coupled to each other and studied in the isotropic regime (all the links between the two polyhedra having the same dynamics). It would be interesting to push the analysis of such a model further and perhaps fully solve it.}

From this perspective, it would seem possible to model the dynamics of a (framed) polyhedron as a matrix model. It would be interesting to see if the tools of matrix models can be relevant in our framework, especially in the large $N$ limit when we would consider the refinement limit of our polyhedra which we expect to describe some continuous 2d surface (topologically equivalent to a 2-sphere).

\section*{Conclusion}

To summarize, we started by explaining how to extend the set of (convex) polyhedra with $N$ faces to a set of framed polyhedra by attaching the extra data of a $\U(1)$ phase to each face. This allows to see the set of framed polyhedra (up to 3d rotations) as the symplectic  quotient $\C^{2N}//\SU(2)$, defined as the set of collections of $N$ spinors satisfying closure constraints and up to $\SU(2)$ transformations. Discussing the various parametrization of this space, we showed that this symplectic manifold is equal to the quotient $\C^{2N}/\SL(2,\C)$, where we can use a $\SL(2,\C)$ transformation to map any collection of spinors onto one satisfying the closure constraints and thus defining a true geometric polyhedron. Furthermore, following the original work of \cite{un1,un2}, the space of framed polyhedra can be identified to the Grassmaniann space $\U(N)/\,(\SU(2)\times\U(N-2))$ with a natural action of the unitary group $\U(N)$ on framed polyhedra. It is important to emphasize that there is no $\U(N)$ action on polyhedra and that the extra phase attached to each face is essential to the construction. These $\U(N)$ transformations allow to generate any framed polyhedra from the totally squeezed configuration with only two non-trivial faces, and thus allow to go between any two framed polyhedra with equal total boundary area.
Such transformations could be instrumental in the study of geometric properties of polyhedron, especially in order to consistently explore the space of polyhedra (either analytically or numerically).

Using this $U(N)$ structure, we have shown how to compute the average value of geometrical observables, such as polynomials in the area of the faces and the angles between their planes (or normal vectors), can be computed as integrals over the unitary group. We have reviewed various formalisms allowing to compute consistently these polynomial integrals over $\U(N)$. Moreover, we have discussed how the Itzykson-Zuber integral can be used as a generating function for these averages. In short, this formula from matrix models contains all the information about the distribution of polyhedra and their shape with respect to the uniform Haar measure on $\U(N)$.

Moving on to the quantum level, we have explained how all the classical features are upgraded automatically upon a canonical quantization of the framed polyhedra phase space. This leads to the Hilbert space of $\SU(2)$-intertwiners (or equivalently $\SU(2)$-invariant states) and one can define semi-classical intertwiner states, that transform coherently under the action of $\U(N)$ and that are peaked on classical framed polyhedra. Then one can compute the trace of polynomial observables. Furthermore, similarly to the classical case, one can use  the character formula for $\U(N)$ group elements as a generating function for these polynomial traces and as a extension of the Itzykson-Zuber formula to the quantum case. We provide two different expressions for the $\U(N)$-character, either as a quotient of generalized Vandermonde determinants or as an Gaussian integral over almost-closed configurations of spinors (using the coherent intertwiner formalism).

\medskip

We also showed how we can describe polygons in a similar fashion, trading the unitary group for the orthogonal group and defining a phase space of ``signed'' polygons as the Grassmanniann space  $\O(N)/\,(\SO(2)\times\O(N-2))$. All the same techniques presented for the unitary group and polyhedra can then be straightforwardly translated to polygons. This lower-dimensional toy models allow to discuss more explicitly the geometrical reconstruction of polygons, which is simpler than for polyhedra, and we plan to investigate in the future the details of gluing these polygons together and the definition of consistent dynamics on the resulting 2d discrete manifolds. Quantizing the system, we would then obtain the dynamics of quantum surfaces. 

The present formalism might also turn out useful in discrete geometry, outside of the realm of quantum gravity, when studying polygons from a purely mathematical point of view. For example, it might be applicable to issues like the search for the largest small polygons, e.g. \cite{octagons}, or other similar problems of geometry.

\medskip

To conclude, we would like to mention a few directions that can be explored based on the present work:
\begin{itemize}

\item After having understood in details all the kinematics on the space of polyhedra (and polygons), we should move to the study of dynamics along the outline shortly discussed earlier in section \ref{outlook}.  In the context of loop quantum gravity, this would mean looking at the dynamics of a fundamental chunk of volume, either at the classical level with a polyhedron or at the quantum level with an intertwiner. In the present framework, it would be most natural to study a deformation dynamics, at fixed number of faces $N$ and fixed total boundary area $\lambda$, with the shape of the polyhedron entirely encoded in the unitary matrix $U$ (up to 3d rotations $\SU(2)$ and action of the stabilizer group $\U(N-2)$). 
It could first be interesting to check what would a free evolution on $\U(N)$, of the type $U[\tau]=\exp(i\tau h)$ for a fixed Hermitian matrix $h$, would produce in terms of polyhedra. Then we could deform such an evolution with a non-trivial potential.
A second step would be to include a non-trivial dynamics for the total area and number of faces, using the $F$-operators, to account for an expansion or contraction of the polyhedron.

From a physical perspective, it would be interesting to relate such dynamics to cosmological mini-superspace models (as attempted in \cite{sfcosmo}) or to quantum black hole models.

\item Here we have developed techniques to compute consistently the average or trace of polynomial observables. We have focused on the area observable, which is well-understood. It would be interesting to apply these methods to a less-understood operator, for instance the volume operator. Indeed the (squared) volume operator is cubic in the normal vectors (or equivalently in the $\su(2)$ generators at the quantum level) and determining its full spectrum is a yet-unsolved problem despite great progress \cite{volumelqg}. There have been a few very interesting approaches to this issue and hopefully we could get some extra information from the $\U(N)$ approach presented here.

\item For now, we have focused on a single polyhedron and then a single intertwiner at the quantum level. The next step is to generalize this to bounded regions of twisted geometries, i.e. to look at a bunch of polyhedra glued together and study their algebra of bulk and boundary deformations. This would be relevant for coarse-graining spin networks in loop quantum gravity and investigate the continuum limit of the theory (or at least, define it more rigorously at the kinematical level). Moreover these deformations should somehow be related to the action of diffeomorphisms on the twisted geometries and spin network states. By a quick glance at the corresponding structures, it appears that it will be possible to describe boundary deformations again by $\U(N)$ transformations, but in different representations than used in the case of the single polyhedron. In this context, it seems plausible to be able to describe the boundary dynamics of spin network states as some matrix models, which at a purely speculative level would open a possibility to a link to a conformal field theory description of the boundary of loop quantum gravity (maybe along the CFT description intertwiners already hinted in \cite{tetrahedron}).

\end{itemize}

\section*{Acknowledgment}

E.L. would like to acknowledge Mehdi Assanioussi for his collaboration on the calculations of entropy and polynomial integrals over the unitary group, as part of his Masters thesis research project at the Laboratoire de Physique ENS Lyon, ``Simple models for black holes in Loop Quantum Gravity'' (June 2012, Aix-Marseille university, France).

\appendix

\section{Computing the density of polyhedra and correlations}
\label{brutal}

We use the method of \cite{counting} to compute the density of polyhedra with $N$ faces and with fixed total area $2\lambda$ and the various averages over the ensemble of such polyhedra. We introduce the following (generating) function:
\be
\label{rho_def}
\rho_{N}[\lambda]
\,\equiv\,
8\pi
\int\prod_i^N\f{d^3\vV_i}{4\pi V_i}\,
\delta\left(\sum_k V_k-2\lambda\right)\,
\delta^{(3)}\left(\sum_k \vV_k\right)\,.
\ee
Following \cite{counting}, we Fourier-transform both sets of constraints and perform the integrals over the normal vectors $\vV_k$ (assuming $\eps>0$)~:
\be
\rho_{N}[\lambda]
\,=\,
8\pi\,e^{2\eps \lambda}
\int \f{dq}{2\pi}\int \f{d^3\vu}{(2\pi)^3}
\int\prod_i^N\f{d^3\vV_i}{4\pi V_i}\,
e^{-\eps V_i}e^{iqV_i}\,
e^{i\vu\cdot\sum_i \vV_i}
\,=\,
8\pi\,e^{2\eps \lambda}
\int \f{dq}{2\pi}\int \f{d^3\vu}{(2\pi)^3}
I(q,\vu)^N
\ee
$$
\textrm{with}\qquad
I(q,\vu)
\,=\,
\int\f{d^3\vV}{4\pi V}\,
e^{-\eps V}e^{iqV}e^{i\vu\cdot\vV}.
$$
The kernel $I(q,\vu)$ converges due to the regulator $\eps>0$. We first integrate over the angular part of $\vV$ and then over its norm:
\be
I(q,\vu)
\,=\,\int_0^{+\infty} V^2 \f{dV}{V}\,e^{-\eps V}e^{iqV}\,
\int_{\cS^2}\f{d^2\hat{V}}{4\pi}\,e^{i\vu\cdot\vV}
\,=\,\int_0^{+\infty} V dV\,e^{-\eps V}e^{iqV}\,\f{\sin uV}{uV}
\,=\,\f{1}{u^2-(q+i\eps)^2}\,.
\ee
This is exactly as the Feynman propagator in quantum field theory where $\vu$ plays the role of the momentum and $q$ the role of the mass.
We then perform the integrals over $\vu$ and finally over $q$:
\beq
\rho_{N}[\lambda]
&=&
8\pi\,e^{2\eps \lambda}\,
\int_{\R}\f{dq}{2\pi}\,e^{-2iq\lambda}\int_0^{+\infty} \f{4\pi u^2 du}{(2\pi)^3}\,
\f{1}{(u^2-(q+i\eps)^2)^N}\nn\\
&=&
e^{2\eps \lambda}\,\int_{\R}\f{dq}{2\pi}\,e^{-2iq\lambda}\,
\f{(2N-4)!}{(N-1)!(N-2)!}\f{(\eps-iq)^{3-2N}}{2^{2N-4}}\nn\\
&=&
\f{\lambda^{2N-4}}{(N-1)!(N-2)!}\,.
\label{rho}
\eeq
And we recover the formula \eqref{rho0} for the volume of the space of (framed) polyhedra with $N$ faces and fixed total area.

\medskip

To extract the average of the norm of a normal vector, or its powers, we can differentiate with respect to $q$. For instance, we have:
\be
\la V_i\ra
\,=\,
\f{8\pi\,e^{2\eps \lambda}}{\rho_{N}}\,
\int \f{d^3\vu}{(2\pi)^3}\,\f{dq}{2\pi}\,
e^{-2iq\lambda}
I(q,\vu)^{N-1}\tI(q,\vu)\,,
\ee
with a modified kernel taking into account the insertion of the observable $V_i$,
$$
\tI(q,\vu)
=
\int \f{d^3\vV}{4\pi V}\,V\,e^{-\eps V}e^{iqV}e^{i\vu\cdot\vV}
=
-i\pp_q I(q,\vu)
=
\f{2(\eps-iq)}{(u^2-(q+i\eps)^2)^2}\,,
$$
which allows to get the expected result:
\be
\la V_i\ra
\,=\,
\f{1}{\rho_{N}}\f{2\lambda^{2N-3}}{N!(N-2)!}
\,=\,
\f {2\lambda}N
\,.
\ee
One can go further and compute averages of higher powers of the norm $V$ by differentiating more. For instance, in order to compute $\la V^2 \ra$, one insert the modified kernel \cite{counting}:
$$
\widetilde{\tI}(q,\vu)
=
\int \f{d^3\vV}{4\pi V}\,V^2\,e^{-\eps V}e^{iqV}e^{i\vu\cdot\vV}
=
-\pp_q^2 I(q,\vu)
=
\f{-2}{(u^2-(q+i\eps)^2)^2}+\f{8(\eps-iq)^2}{(u^2-(q+i\eps)^2)^3}\,.
$$
Using this technique, one can also compute the correlations between the norm of two different faces, by modifying the product of kernels in the integral from $I^N$ to $I^{N-2}\tI^2$.

Similarly, one differentiate with respect to the components of the vector $\vu$ in order to insert observables depending on the components of the normal vectors $\vV$. For instance, we can compute the correlations $\la V_i^aV_j^b\ra$ as:
\be
\rho^{ab}_{i}
\,\equiv\,
\la V_i^aV_i^b\ra
\,=\,
\f{8\pi\,e^{2\eps \lambda}}{\rho_{N}}\,
\int \f{d^3\vu}{(2\pi)^3}\,\f{dq}{2\pi}\,
e^{-2iq\lambda}
I(q,\vu)^{N-1}(-\pp_{u_b}\pp_{u_a}I(q,\vu))
\,=\,
2\delta^{ab}\f{\lambda^2}{N(N+1)}\,,
\ee
\be
\rho^{ab}_{i\ne j}
\,\equiv\,
\la V_i^aV_j^b\ra
\,=\,
\f{8\pi\,e^{2\eps \lambda}}{\rho_{N}}\,
\int \f{d^3\vu}{(2\pi)^3}\,\f{dq}{2\pi}\,
e^{-2iq\lambda}
I(q,\vu)^{N-2}(-i\pp_{u_b}I(q,\vu))(-i\pp_{u_a}I(q,\vu))
\,=\,
-2\delta^{ab}\f{\lambda^2}{N(N^2-1)}\,.
\ee
One can check as a consistency check that $\sum_{i,j}\la V_i^aV_j^b\ra=0$ as expected from the closure constraints.

Furthermore, one interesting observable is the second moment $\Theta^{ab}=\sum_i V_i^aV_i^b-\f13\delta^{ab}V_iV_i$, which characterizes the shape of the polyhedra and its deviation from the isotropic spherical distribution. From the result above, we have the trivial average $\la \Theta^{ab} \ra=0$. However, using the technique presented here, one can compute the width(s) or uncertainty of the distribution of the tensor $\Theta^{ab}$ around its vanishing mean value. This is rather lengthy calculation although straightforward, which we do not detail here. The final result is:
\be
\la \Theta^{ab}\Theta^{cd} \ra
\,=\,
\lambda^4\,
\f{4\left(4(N^2+N-2)\delta^{ab}\delta^{cd}-6(N-1)(\delta^{ac}\delta^{bd}+\delta^{ad}\delta^{bc})\right)}{3(N-1)N(N+1)(N+2)(N+3)}
\,.
\ee



\begin{thebibliography}{99}

\bibitem{lqg_review}
C. Rovelli,
{\it Zakopane lectures on loop gravity},
arXiv:1102.3660;\\
P. Doná \& S. Speziale,
{\it Introductory lectures to loop quantum gravity},
 arXiv:1007.0402;\\
H. Sahlmann,
{\it Loop quantum gravity - a short review},
arXiv:1001.4188;\\
M. Gaul and C. Rovelli,
{\it Loop Quantum Gravity and the Meaning of Diffeomorphism Invariance},
Lect.Notes Phys. 541 (2000) 277-324 [arXiv:gr-qc/9910079]
 

\bibitem{spinor1}
E.F. Borja, L. Freidel, I. Garay, E.R. Livine,
{\it U(N) tools for Loop Quantum Gravity: The Return of the Spinor},
Class.Quant.Grav.28 (2011) 055005 [arXiv:1010.5451]

\bibitem{spinor2}
E.R. Livine and J. Tambornino,
{\it Spinor Representation for Loop Quantum Gravity},
J. Math. Phys. 53, 012503 (2012) [arXiv:1105.3385];\\
E.R. Livine and J. Tambornino,
{\it Loop gravity in terms of spinors},
Proceedings of Loops '11 Conference (Madrid, 2011) [arXiv:1109.3572]

\bibitem{spinor3}
E.R. Livine and J. Tambornino,
{\it Holonomy Operator and Quantization Ambiguities on Spinor Space},
to appear in Phys.Rev.D 2013 [arXiv:1302.7142]

\bibitem{twisted1}
L. Freidel and S. Speziale,
{\it Twisted geometries: A geometric parametrisation of SU(2) phase space},
Phys.Rev.D82 5 (2010) 084040 [arXiv:1001.2748]

\bibitem{twisted2}
L. Freidel and S. Speziale,
{\it From twistors to twisted geometries},
Phys.Rev.D82 (2010) 084041 [arXiv:1006.0199]

\bibitem{un1}
L. Freidel and E.R. Livine,
{\it The Fine Structure of SU(2) Intertwiners from U(N) Representations},
J.Math.Phys.51 (2010) 082502 [arXiv:0911.3553]


\bibitem{KM}
M. Kapovich and J.J.Milson,
{\it The Symplectic Geometry of Polygons in Euclidean space},


\bibitem{un2}
L. Freidel and E.R. Livine,
{\it U(N) Coherent States for Loop Quantum Gravity},
J. Math. Phys. 52 (2011) 052502 [arXiv:1005.2090]

\bibitem{iz1}
C. Itzykson and J.B. Zuber,
{\it The planar approximation II.},
J. Math. Phys. 21 (1980) 411

\bibitem{dona}
E. Bianchi, P. Dona and S. Speziale,
{\it Polyhedra in loop quantum gravity},
Phys.Rev.D83 (2011) 044035 [arXiv:1009.3402]


\bibitem{un3}
M. Dupuis and E.R. Livine,
{\it Holomorphic Simplicity Constraints for 4d Spinfoam Models},
 Class. Quant. Grav. 28 (2011) 215022 [arXiv:1104.3683]

\bibitem{un4}
M. Dupuis and E.R. Livine,
{\it Revisiting the Simplicity Constraints and Coherent Intertwiners},
Class.Quant.Grav.28 (2011) 085001 [arXiv:1006.5666]



\bibitem{un0}
F. Girelli and E.R. Livine,
{\it Reconstructing Quantum Geometry from Quantum Information: Spin Networks as Harmonic Oscillators},
Class.Quant.Grav. 22 (2005) 3295-3314 [arXiv:gr-qc/0501075]



\bibitem{simplicityL}
M. Dupuis, L. Freidel, E.R. Livine and S. Speziale,
{\it Holomorphic Lorentzian Simplicity Constraints},
J. Math. Phys. 53 (2012) 032502  [arXiv:1107.5274]

\bibitem{closure}
F. Conrady and L. Freidel,
{\it Quantum geometry from phase space reduction},
J.Math.Phys.50 (2009) 123510 [arXiv:0902.0351]

\bibitem{tetrahedron}
L. Freidel, K. Krasnov and E.R. Livine,
{\it Holomorphic Factorization for a Quantum Tetrahedron},
Commun.Math.Phys.297 (2010) 45-93 [arXiv:0905.3627]


\bibitem{generating}
V. Bonzom and E.R. Livine,
{\it Generating Functions for Coherent Intertwiners},
Class. Quantum Grav. 30 (2013) 055018 [arXiv:1205.5677]

\bibitem{hal_volume}
E. Bianchi and H.M. Haggard,
{\it Discreteness of the volume of space from Bohr-Sommerfeld quantization},
Phys.Rev.Lett.107 (2011) 011301 [arXiv:1102.5439];\\
E. Bianchi and H.M. Haggard,
{\it Bohr-Sommerfeld Quantization of Space},
 arXiv:1208.2228

\bibitem{eugenio_bh}
E. Bianchi,
{\it Black Hole Entropy, Loop Gravity, and Polymer Physics},
Class.Quant.Grav.28 (2011) 114006 [arXiv:1011.5628]


\bibitem{counting}
E.R. Livine and D.Terno,
{\it Entropy in the Classical and Quantum Polymer Black Hole Models},
Class. Quantum Grav. 29 (2012) 224012 [arXiv:1205.5733]

\bibitem{lqg_bh_su2}
J. Engle, K. Noui, A. Perez and D. Pranzetti,
{\it Black hole entropy from an SU(2)-invariant formulation of Type I isolated horizons},
 arXiv:1006.0634;\\
J. Engle, K. Noui, A. Perez and D. Pranzetti,
{\it The SU(2) Black Hole entropy revisited},
JHEP 1105 (2011) 016 [ arXiv:1103.2723]

\bibitem{lqg_bh_review}
R. Kaul,
{\it Entropy of Quantum Black Holes},
SIGMA 8 (2012) 005 [arXiv:1201.6102]; \\
J. Diaz-Polo and D. Pranzetti,
{\it Isolated Horizons and Black Hole Entropy in Loop Quantum Gravity},
SIGMA 8 (2012) 048 [arXiv:1112.0291]

\bibitem{hayden}
P. Hayden, D.W. Leung and A. Winter,
{\it Aspects of generic entanglement},
Comm. Math. Phys. 265 (2006) 1, 95-117 [arXiv:quant-ph/0407049]


\bibitem{collins}
B. Collins,
{\it Moments and Cumulants of Polynomial random variables on unitary groups, the Itzykson-Zuber integral and free probability},
Int. Math. Res. Not. (17) (2003) 953-982 [math-ph/0205010]

\bibitem{brandao}
F.G.S.L. Brand\~ao, P. C'wiklin'ski, M. Horodecki, P. Horodecki, J. Korbicz, and M. Mozrzymas,
{\it Convergence to equilibrium under a random Hamiltonian},
Phys. Rev. E 86 (2012) 031101  [arXiv:1108.2985]


\bibitem{spengler}
C. Spengler, M. Huber and B.C. Hiesmayr,
{\it A composite parameterization of unitary groups, density matrices and subspaces},
J. Phys. A: Math. Theor. 43 (2010) 385306   [arXiv:1004.5252];\\
C. Spengler, M. Huber and B.C. Hiesmayr,
{\it Composite parameterization and Haar measure for all unitary and special unitary groups},
J. Math. Phys. 53 (2012) 013501     [ arXiv:1103.3408]

\bibitem{spengler2}
C. Spengler,  M. Huber, A. Gabriel and B.C. Hiesmayr,
{\it Examining the dimensionality of genuine multipartite entanglement},
Quantum Inf. Process. 12 (2013) 269     [arXiv:1106.5664];\\
C. Spengler,  M. Huber, S. Brierley, T. Adaktylos and B.C. Hiesmayr,
{\it Entanglement detection via mutually unbiased bases},
Phys. Rev. A 86 (2012) 022311       [arXiv:1202.5058]

\bibitem{hamiltonian3d}
V. Bonzom and E.R. Livine,
{\it A new Hamiltonian for the Topological BF phase with spinor networks},
J. Math. Phys. 53 (2012)  072201 [arXiv:1110.3272]

\bibitem{intertwiner}
E.R. Livine and S. Speziale,
{\it A new spinfoam vertex for quantum gravity},
Phys.Rev.D76 (2007) 084028 [arXiv:0705.0674]

\bibitem{volumelqg}
J. Brunnemann and T. Thiemann,
{\it Simplification of the Spectral Analysis of the Volume Operator in Loop Quantum Gravity},
Class.Quant.Grav. 23 (2006) 1289-1346 [arXiv:gr-qc/0405060]; \\
K.A. Meissner,
{\it Eigenvalues of the volume operator in loop quantum gravity},
Class.Quant.Grav. 23 (2006) 617-626 [arXiv:gr-qc/0509049]
E. Bianchi and H. Haggard,
{\it Discreteness of the volume of space from Bohr-Sommerfeld quantization},
hys.Rev.Lett.107 (2011) 011301 [arXiv:1102.5439]

\bibitem{sfcosmo}
E.R. Livine and M. Martin-Benito,
{\it Classical Setting and Effective Dynamics for Spinfoam Cosmology},
Class. Quantum Grav. 30 (2013) 035006 [arXiv:1111.2867]

\bibitem{octagons}
C. Audeta, P. Hansenb, F. Messined and J. Xionge,
{\it The Largest Small Octagon},
Journal of Combinatorial Theory, Series A,  98-1 (2002) 46?59,
http://dx.doi.org/10.1006/jcta.2001.3225

\end{thebibliography}
\end{document}